\documentclass[runningheads]{llncs}
\usepackage{graphicx}
\usepackage{stmaryrd} 
\usepackage{amsmath}
\usepackage{lineno}
\usepackage{listings}
\usepackage{ebproof} 
\usepackage{lineno}
\begin{document}
\title{Execution Time Program Verification With Tight Bounds}
%
%
\author{Ana Carolina Silva\inst{1} \and
Manuel Barbosa\inst{1,2} \and
Mário Florido\inst{1,3}}
\authorrunning{A. C. Silva et al.}
%
\institute{FCUP, Universidade do Porto \and
INESC TEC \and
LIACC, Universidade do Porto}
\maketitle              
\begin{abstract}
This paper presents a proof system for reasoning about execution time bounds for a core imperative programming language. Proof systems are defined for three different scenarios: approximations of the worst-case execution time, exact time reasoning, and less pessimistic execution time estimation using amortized analysis. We define a Hoare logic for the three cases and prove its soundness with respect to an annotated cost-aware operational semantics. Finally, we define a verification conditions (VC) generator that generates the goals needed to prove program correctness, cost, and termination. Those goals are then sent to the Easycrypt toolset for validation. The practicality of the proof system is demonstrated with an implementation in OCaml of the different modules needed to apply it to example programs. Our case studies are motivated by real-time and cryptographic software.

\end{abstract}
\section{Introduction}\label{sec:introduction}
Semantics-based approaches to program verification usually belong to two different broad classes: 
1) partial correctness assertions expressing relations between the initial and final state of program variables in the form of pre and postconditions, assuming that the program terminates; and 2) total correctness properties, which besides those assertions which specify claims about program behavior, also express program termination.

However, another class of properties is fast growing in relevance as a target for program verification: resource consumption when executing a program. The term {\em resource}  is used broadly: resources can be time used to execute the program on a particular architecture, memory used (stack or heap) during program execution, or even energy consumption. Resource consumption has a significant impact in different specific areas, such as real-time systems, critical systems relying on limited power sources, and the analysis of timing side-channels in cryptographic software. 

A proof system for total correctness can be used to prove that a program execution terminates, but it does not give any information about the resources it needs to terminate. In this dissertation, we want to study extended proof systems for proving assertions about program behavior that may refer to the required resources and, in particular, to the execution time.

Proof systems to prove bounds on the execution time of program execution were defined before in \cite{nielson1987,barbosa21}. Inspired by the work presented in \cite{barbosa21} our goal is to study inference systems that allow proving assertions of the form $\{ \varphi \} C \{\psi | t\}$, meaning that if the execution of the statement $C$ is started in a state that validates the precondition $\varphi$ then it terminates in a state that validates postcondition $\psi$ and the required execution time is at most of magnitude $t$.

\section{Goals and Contributions}\label{sec:intro-obj}

Our main goal is to define an axiomatic semantics-based proof system for reasoning about execution time bounds for a core imperative programming language. Such a system would be useful not only to understand the resource necessities of a program but also it could be applied to cryptographic implementations to prove the independence of resource usage from certain program variables.
This high-level goal translates into the following concrete objectives:
\begin{enumerate}
    \item Study axiomatic systems. This will further our knowledge in the field and allow us to understand how to define our logic for resource analysis.
    \item Study amortized analysis so we can understand how to apply amortization to a proof system to refine cost-bound estimation of while loops.
    \item Analyze the state-of-the-art to understand what has already been developed to analyze resource consumption and the main limitations found.
    \item Develop a sound logic capable of verifying correction, terminations, and bounds on resource consumption using a simple imperative programming language.
    \item Create a tool based on our logic, capable of verifying time bounds, correction, and terminations, for example-problems.
    \item Apply this logic to analyze the time complexity of classic algorithms.
\end{enumerate}

The main contribution of this paper is a proof system that can verify resource assumptions in three different scenarios:
\begin{enumerate}
\item Upper bounds on the required execution time. This is mostly an adaptation of previous work in \cite{barbosa21}.
\item Amortized costs denoting less pessimistic bounds on the execution time.
\item Exact costs for a fragment of the initial language with bounded recursion and a constrained form of conditional statements.
\end{enumerate}

The two last scenarios are a novel contribution of our system, and we treat them in a unified way to enable their integrated use. 

Assertions on program behavior that establish upper bounds on execution time may be useful for general programming, where one wants to prove safety conditions concerning the worse case program complexity. 
As in prior approaches, the tightness of the bound is not captured by the logic, and there is often a trade-off between the tightness of the proved bound and the required proof effort.

Proofs that leverage amortized costs may be used when trivially composing worst-case run-time bounds results in overly pessimistic analyses. This is particularly useful for algorithms where some components imply a significant cost in resources, whereas other components are not as costly. With amortized costs, we may prove assertions about the aggregate use of costly and less costly operations over the whole algorithm execution. 

Finally, the third class of assertions denoting exact costs are useful in scenarios where the approximation of execution time is not enough to guarantee safety, as it happens for critical systems and real-time programming. 
Moreover, proving that the exact execution time of a program is an expression that does not depend on confidential data provides a direct way to prove the absence of timing leakage, which is relevant in cryptographic implementations.
We must restrict the programming language to guarantee the ability to prove exact costs. Thus, in this third scenario, programs have bound recursion, and conditional statement branches have to have the same cost.

Before defining our proof system, we defined an operational semantics capable of computing the execution time for expressions and statements during program execution. This cost-aware operational semantics is another contribution of our work, and it is used to prove the soundness of our inference system.

A third contribution of this work, which shows the practicality of our proof system, is an implementation in OCaml of the different modules needed to apply it to example programs. We then present several application examples motivated by real-time and cryptographic software.

\section{Document Structure}\label{sec:intro-structure}
This document is organized as follows:
\begin{itemize}
    \item The first chapter - \textbf{Introduction} - gives a context of our work in the field, the motivation for this project, our main goals, our contributions, and how the document is organized.
    \item The second chapter - \textbf{Background} - elaborates on the theoretical results used in the basis of our work and needed to understand our definitions and results.
    \item The third chapter - \textbf{Related Work} - presents an analysis of the literature on static resource analysis, from type-based systems to axiomatic semantics systems.
    \item The fourth chapter - \textbf{Cost Aware Program Logic} - presents our language definition, our original logic for upper bound estimation, the respective VCG, and some illustrative examples.
    \item The fifth chapter - \textbf{Amortized Costs} - briefly introduces the field of amortized analysis and presents an extension to the logic and VCG from chapter 4, with the use of amortized analysis to improve the upper-bound estimation.
    \item The sixth chapter - \textbf{Exact Logic} - presents a variation to our language and an extension to our logic that allows for the derivation of the exact cost of a program.
    \item The seventh chapter - \textbf{Implementation and Experimental Results} - shows the architecture of the tool developed, as well as some implementation details and practical results.
    \item The eighth chapter - \textbf{Conclusion and Future Work} - reflects on the main conclusions from our research and developed work and presents some goals to further extend and improve our project.
\end{itemize}

\section{Background}\label{sec:background}
In this chapter, we provide an overview of the field of formal verification and some of the most relevant theoretical results that are the basis of our work.
We start by giving historical background on the field of formal verification. Here we will present results, such as the ones achieved by Floyd and Hoare, used as the base of our definitions. We will also define concepts fundamental to understanding the work presented in this dissertation.

\section{Historical Background}\label{sec:pre-history}
As computers became more powerful,  programs also became longer and more complex. When programs were still relatively small, flowcharts or extensive testing was enough to prove a program's functionality. 
But programs quickly started being so complex that these methods became unreliable and more prone to error. 

At the beginning of the second half of the 20th-century, experts started to find vulnerabilities in public distributed software. 
Since then, the use of computational systems has grown exponentially, and so did the number of vulnerabilities and their impact. A simple error might have drastic consequences, such as a leak of confidential information, the crash of critical systems, and direct loss of assets.



This problem proved to be enough reason to start thinking about a more reliable way to guarantee the properties of a program and develop tools that help verify these properties. 

Formal Verification refers to using mathematical principles to prove the correction of a given specification of a program. 
It is hard to pinpoint where it all started, but the works of Robert Floyd~\cite{floyd1967} and Tony Hoare~\cite{hoare1969} were undoubtedly pioneers in the field, and their definitions are still the base of verification tools used today.


\section{Semantics}\label{sec:pre-semantics}

When defining programming languages, we want a way to be capable of reasoning about what programs are doing. 
The syntax describes the grammatical rules we must follow to write a program in a language. The syntax allows us to distinguish between languages and identify a program's language.
But if we want to understand what that program is doing, we need to look at its semantics. Semantics is a way to make sense of the meaning of a program and understand what it is trying to accomplish.

There are multiple strategies to analyze the meaning of a program. The most popular ones are operational, denotational, and axiomatic semantics. 
Operational semantics focus on what steps we take during the program's execution. In denotational semantics, we do not care about the "how" but only about "what" the program is doing. In axiomatic semantics, we are concerned about evaluating the satisfability of assertions on the program and its variables. We will go more in-depth on how operational and axiomatic semantics work.

\subsection{Operational Semantics}\label{subsec:operational}

Operational semantics describes the meaning of a program by specifying the transitions between states of an abstract state machine.
As we mentioned, unlike with denotational semantics, here we are concerned about \emph{how} the machine changes states with the execution of a statement.

There are two styles of operational semantics
\begin{itemize}
    \item Small-step or Structural Operational Semantics
    \item Big-step or Natural Semantics
\end{itemize}

In Structural Operational Semantics or Small-step Semantics, we are concerned about every individual transition we take throughout the program's execution.
In Natural Semantics or Big-step semantics, we want to understand how we transition from the initial to the final state. We are concerned about a high-level analysis of how the machine state changes and not about each individual step.

\subsection{Axiomatic Semantics (Hoare Logic)}\label{subsec:pre-hoare}
In 1967 Floyd specified a method that would allow proving properties on programs, such as correctness, equivalence, and termination \cite{floyd1967}. They achieved this by representing programs as flowcharts and associating propositions to each connection on the flowchart. The proof is done by induction on the number of steps. If an instruction is reached by a connection whose proposition is true, then we must leave it with a true condition as well. 
In 1969 Hoare wrote a paper where they extended Floyd's logic to prove properties on a simple imperative program \cite{hoare1969}. In this paper, they defined what we now call Hoare (or Floyd-Hoare) triples.
\begin{definition}[Hoare Triples]\label{def:hoare-triple}
A Hoare triple is represented as $$\{P\}Q\{R\}$$ 
and can be interpreted as "if the assertion $P$ is true before we run program $Q$, then assertion $R$ will be true when the program ends".
\end{definition}
Notice this definition does not offer guarantees over termination. Executing program $Q$ from a state validating $P$, does not have to halt. As long as whenever it does the final state validates $R$. We call assertions in the form $\{P\}Q\{R\}$ \emph{partial correctness assertions}.

 Using this definition Hoare specifies a proof system with a set of axioms and inference rules, which allow to prove assertions on any program written in this language. A derivation on this proof system is called a \emph{theorem} and is written as $\vdash \{P\}Q\{R\}$.

As an example, let us look at the assignment axiom. Consider an assignment to variable $x$ of an expression $a$ in the form 
$$x := a$$

If an assertion P is true after executing the assignment (when variable x takes the value of expression a) then it has to be true before the assignment if we replace any mentions of $x$ in P by $a$. This is usually represented as $P[a/x]$.
Therefore the assignment axiom is written as
$$\{P[a/x]\}\ x := a\ \{ P \}$$

If, in addition to proving a program specification is correct, we also want to prove the program always halts, we are looking for Total Correctness.

\begin{definition}[Total Correctness]
A total correctness assertion is represented as
$$[P]Q[R]$$
where $P$ and $R$ are assertions and $Q$ is a program. If we execute $Q$ from a state that satisfies $P$ program Q will terminate and the final state will satisfy $R$.
\end{definition}

partial correctness + termination = total correctness

Total correctness is harder to prove than partial correctness and not always possible. But it also gives a stronger guarantee about a programs behavior.

We consider two properties on this proof system, soundness an completeness. Soundness ensures our proof systems generates valid partial correctness assertions. Completeness ensures that our system is capable of deriving every valid assertion.

\begin{definition}[Validity]
We say an assertion \{P\}Q\{R\} is valid if it is true for all possible states.
$$\forall \sigma \in \Sigma_\perp.\ \sigma \models \{P\}Q\{R\}$$
Or we can simply represent it as
$$\models \{P\}Q\{R\}$$
\end{definition}

\begin{definition}[Soundness]
Our proof system is sound if every rule preserves validity. 
In other words, for any partial correctness assertion $\{P\}Q\{R\}$
$$\mathrm{if}\ \vdash \{P\}Q\{R\}\ \mathrm{then}\ \models \{P\}Q\{R\}$$
\end{definition}

\begin{proof}
Soundness is proved by structural induction on the statement $Q$.
\end{proof}

\begin{definition}[Completeness]
A proof system is complete if every true assertion $\{P\}Q\{R\}$ can be proved by our system.
$$\mathrm{if}\ \models \{P\}Q\{R\}\ \mathrm{then}\ \vdash \{P\}Q\{R\}$$
\end{definition}

Proving completeness is not as trivial and most of the times not possible. The completeness of the proof system presented by Hoare was established by Cook in 1978~\cite{cook1978}. In this paper he presents a proof of \textit{relative completeness}, that is, assuming our assertion language is complete then the logic presented by Hoare is also complete. 


\section{Verification Conditions Generator}

We now have the necessary notation to specify the behavior of a program. Hoare's set of axioms and rules allows proving that program's said behavior. 
Manually proving these properties is not only long and tedious but also prone to error. We are missing a mechanized solution we could apply to any program to guarantee its validity.

Dijkstra defined the weakest precondition algorithm in is 1976 book "A Discipline of Programming" \cite{dijkstra1976}.

\begin{definition}[Weakest Precondition]
The weakest precondition is the simplest condition, necessary and sufficient to guarantee the post-condition is true when the program terminates. We use the notation $wp(Q,R)$ where Q is a statement and R is a post-condition. 
$$\models \{P\}Q\{R\}\ \textit{iff}\ P \rightarrow wp(Q,R)$$

\end{definition}

Let us consider as an example the statement $Q \equiv x := y + 2$ and we want to prove $R \equiv x \ge 0$ is true after executing Q. The weakest precondition $wp(Q,R)$ would say that as long as $y \ge 2$ is true before executing statement Q, then R will be true after execution.

In 1979, JC King presents the first mechanized algorithm to automatically verify the correctness of a program~\cite{king1970}. His work was based on the definitions provided by Floyd~\cite{floyd1967} and Manna~\cite{manna1969}. 
In more recent implementations this algorithm is usually based on Hoare's definition of correctness and it is called a VCG. This algorithm makes use of the weakest precondition function. Given a Hoare triple $\{P\}Q\{R\}$, for the program Q to be correct we must guarantee that $P \rightarrow wp(Q,R)$ is true. This is the condition that needs to be proved in order to guarantee correctness. If our language contains loops we need to satisfy extra conditions, including the preservation of the loop invariant. 
The \textit{loop invariant} is an assertion that is satisfied before and after every execution of the loop's body.

\begin{definition}[Verification Condition Generator]
A VCG is an algorithm that when applied to a Hoare triple returns a set of VC. The Hoare triple is derivable in our proof system, if and only if all the generated conditions are valid.
$$\models VCG(\{P\}Q\{R\})\ \textit{iff}\ \vdash \{P\}Q\{R\}$$
\end{definition}

\begin{theorem}[Soundness of VCG]\label{theor:soundness-vcg}
$\models VCG(\{P\}Q\{R\})\ \rightarrow\ \vdash \{P\}Q\{R\}$.
\end{theorem}
\begin{proof}
By induction on the structure of Q.
\end{proof}

\begin{theorem}[Completness of VCG]\label{theor:completeness-vcg}
$\vdash \{P\}Q\{R\}\ \rightarrow\ \models VCG(\{P\}Q\{R\})$.
\end{theorem}
\begin{proof}
By induction in the derivation of $\vdash \{P\}Q\{R\}$.
\end{proof}

Returning to our last example where $wp(x:=y+2, x \ge 0) = y \ge -2$ let us consider the partial correctness assertion $\{y = 0\} x := y + 2 \{x \ge 0\}$. Since in this case the program does not contain any loops there is only one verification condition that needs to be satisfied in order to prove correctness: $(y = 0) \rightarrow (y \ge -2)$. Since this VC is valid, we know our program is correct.

\section{Related Work}\label{sec:related}
There has been increasing interest in the field of static resource analysis. Knowing bounds on resources or rough estimates of resource consumption can help us optimize embedded software and real-time systems. 

This chapter briefly describes some of the most relevant work on resource estimation. We divide the chapter into sections according to the methodology used to prove or infer resource bounds. We also compare the literature with the work presented in this dissertation and show the relevance of the work we developed. 

\section{Axiomatic Semantics Systems}

One way of proving bounds on a program
's resources is by using axiomatic semantics. Other works have already implemented systems that use axiomatic semantics for resource analysis, which differ from our work in multiple ways, from the paradigm of language used to the precision of the derived bounds. This section explains some of the most relevant work in the literature that inspired our definitions. We further subdivided this literature into two categories: classical cost analysis and amortized cost analysis.

\subsection{Classical Cost Analysis}
There is already some work using derivations on the logic presented by Hoare~\cite{hoare1969} for resource analysis. Some estimate orders of magnitude, while others use a more detailed annotation to automate the process but lack ways of optimizing the bounds. 

One of the first and still one of the most relevant works on this topic is the one presented by Nielsen~\cite{nielson1987,nielson2007}. The author defines an axiomatic semantics for a simple imperative programming language in this work and extends Hoare's logic so that the proof system would be capable of proving the magnitude of worst-case running time and termination. This system is also proven sound. Even though this work operates on a similar imperative language to the one we defined, it lacks the precision our logic provides since it only allows proving order-of-magnitude.

In 2014, Carbonneaux \emph{et al.}~\cite{carbonneaux2014} presented a system that verifies stack-space bounds of compiled machine code (x86 assembly) at the C level. That is, it derives bounds during compilation from C to assembly. They developed a quantitative Hoare logic capable of reasoning about resource consumption. This work is an extension of the CompCert C Compiler~\cite{leroy2009}. Coq was used to implement and verify the compiler. The work by Carbonneaux focuses on the compilation from C to assembly using quantitative logic, which does not serve the same purpose we are trying to achieve. With our work, we can prove tight bounds on imperative programs using an assisted proof system, where the user can help make the bounds as precise as possible. Also, how default constructors' costs are defined makes our work easy to adjust for a system with different resource usage.

In 2018, Kaminski \emph{et al.} defined a conservative extension of Nielsen's logic for deterministic programs~\cite{nielson1987,nielson2007} by developing a Weakest Precondition calculus for expected runtime on probabilistic programs~\cite{kaminski2018}. Again this work is largely automated, which differs from our user-assisted approach. Since it reasons about probabilistic programs, it faces other challenges than the ones we are interested in this work. 

In 2021 a paper was released extending the logic of EasyCrypt~\cite{barbosa21}. EasyCrypt is an interactive tool created to prove the security of cryptographic implementations. One of its core concepts is a set of Hoare Logics, which allow proofs on relational procedures and probabilistic implementations. In this paper, the authors propose an extension to the EasyCrypt tool, allowing to prove properties on the cost of a program. To achieve this, they extended the existing logic to include cost rules. They also implemented a way to define the cost of custom operators. Our work operates on a subset of EasyCrypt's language, but we extended the logic to use the potential method of amortized analysis, increasing the accuracy of the generated bounds.

\subsection{Amortized Cost Analysis}
We will now present some of the literature that, in addition to using axiomatic semantics for static cost analysis, also uses amortized analysis to increase the accuracy of the bounds.

Carbonneaux \emph{et al.}~\cite{carbonneaux2015} continued their previous work~\cite{carbonneaux2014} on deriving worst-case resource bounds for C programs, but they now implemented a system that uses amortized analysis and abstract interpretations.

In~\cite{haslbeck2018}, Haslbeck and Nipkow analyze the works of Nielson\cite{nielson2007}, Carbonneaux \emph{et al.}~\cite{carbonneaux2014,carbonneaux2015} and Atkey~\cite{atkey2011} and prove the soundness of their systems. In this paper, they implement Verification Condition Generators based on Nielsen's logic and Carbonneaux's Logic, proving it sound and complete. They compare all three methodologies and explain some of the limitations of these systems.

\section{Type-Based Systems}
While our system uses a Hoare logic to prove upper bounds on program cost, many existing systems are type-based. Usually, these systems use type inference and type size/cost annotation in order to be able to analyze resource usage statically. Even though these works highly differ from ours, we will briefly mention some of the most relevant works in the field.

In \cite{radicek2018}, the authors present a proof system for cost analysis on functional programs using a fine-grained program logic for verifying relational and unary costs of higher-order programs. The paper~\cite{avanzini2017} presents a fully automated methodology for complexity analysis of higher-order functional programs based on a type system for size analysis with a sound type inference procedure. Hoffman and Jost, \cite{hoffmann2006} defined a type system capable of analyzing heap space requirements during compilation time based on amortized analysis. This work was limited to linear bounds in the size of the input, so they later extended it to polynomial resource bounds~\cite{hoffmann2010}.

In 2017, Hoffman \emph{et al.}~\cite{hoffmann2017} developed a resource analysis system capable of proving worst-case resource bounds. This resource is a user-defined input and can be anything from time or memory to energy usage. This work is an extension of their previous work in Automatic Amortized Resource Analysis (AARA), where they used amortized analysis to derive polynomial bounds for the first time. Their proof system is a type system with inductive type, refined from OCaml's type system. 

Atkey~\cite{atkey2011} presents a type-based amortized resource analysis system adapted from Hofmann's work to imperative pointer-manipulating languages. They achieve this by implementing a separation logic extension to reason about resource analysis.

Serrano \emph{et al.}~\cite{serrano2014} introduced a general resource analysis for logic programs based on sized types, i.e., types that contain structural information and lower and upper bounds on the size of the terms. They achieved this by using an abstract Interpretation Framework. 

In~\cite{simoes2012}, the authors develop a type-based proof system capable of automatically and statically analyzing heap allocations. This work is an extension of Hoffman's work for a lazy setting.
Vasconcelos \emph{et al.}~\cite{vasconcelos2015} defined a type system capable of predicting upper bounds on the cost of memory allocation for co-recursive definitions in a lazy functional language. 

\section{Other Proof Systems}
Some other works on static estimation of resource bounds use other methodologies other than axiomatic semantics or type theory.
In a 2009 paper, the tool COSTA is presented~\cite{albert2009}. COSTA is a static analyzer for Java bytecode. It infers cost and termination information. It takes a cost model for a user-defined resource as input and obtains an upper bound on the execution time of this resource. 

In~\cite{gulwani2009} they compute symbolic bounds on the number of statements a procedure executes in terms of its input. They use the notion of counter variables and an invariant generation tool to compute linear bounds on these counters. These bounds are then composed to generate a total non-linear bound.

Brockschmidt \emph{et al.}~\cite{brockschmidt2014} uses Polynomial Rank Functions (PRF) to compute time bounds. Then these bounds are used to infer size bounds on program variables. They consider small parts of the program in each step and incrementally improve the bound approximation.

\section{Cost-Aware Program Logic}\label{sec:logic}
In this chapter, we will focus on our first goal of formally verifying the worst-case execution time of imperative programs. To achieve this, we specify a simple while language with annotations. We first define an operational semantics capable of computing the execution time. Having this, we define a cost logic, which we use to prove correctness, termination, and worst-case execution time. In the last section of this chapter, we present a Verification Condition Generator (VCG) algorithm for our logic. Practical results using this algorithm definition are shown in chapter~\ref{sec:results}.

\section{Annotated While Language}\label{chap4:sec:language} 
We start by defining a core imperative language (\textbf{IMP}) with the following syntactic structures: numbers, booleans, identifiers, arithmetic expressions, boolean expressions, statements, and assertions. 
To simplify the presentation and explanation of these structures we will use meta-variables to refer to elements in these sets: $n$ for numbers, $x$ for identifiers, $a$ for arithmetic expressions, $b$ for boolean expressions, $S$ for statements, and $P, Q$ for assertions.
 
\paragraph{Numbers and Booleans} 
We consider numbers, to be the usual set of signed decimal numerals for positive and negative integer numbers. Our boolean set is defined as $\{ true, false\}$.

\paragraph{Arithmetic Expressions}
Arithmetic expressions are defined by the following rules
$$a ::= n\ 
    |\ x\ 
    |\ x [ a ]\ 
    |\ a_1 + a_2\ 
    |\ a_1 - a_2\  
    |\ a_1 * a_2\  
    |\ a_1 / a_2\  
    |\ {a_1}^{a_2}\  
    |\ \sum_{x = n}^{a_1} a_2
$$

\paragraph{Boolean Expressions}
Similarly to arithmetic expressions, booleans expressions are defined as
$$
    b ::= true\ |\ false\ 
    |\ b_1 = b_2 \ 
    |\ b_1 \ne b_2 \ 
    |\ b_1 < b_2 \ 
    |\ b_1 > b_2 \ 
    |\ b_1 \le b_2 \ 
    |\ b_1 \ge b_2 \ 
    |\ \neg b \ 
    |\ b_1 \land b_2 \ 
    |\ b_1 \lor b_2
$$

\paragraph{Statements}
Finally, we define the following statement rules:
$$
S ::= x = a\  
 |\ x [ a_1 ] = a_2\ 
 |\ \textbf{if}\ b\ \textbf{then}\ S_1\ \textbf{else}\ S_2\ \textbf{done}\ 
 |\ \textbf{while}\ b\ \textbf{do}\ S\ \textbf{done}\ 
 |\ S_1 ; S_2
$$

An example of a simple program in our language is shown in figure~\ref{fig:division}, where we note the time annotation on the right-hand side of the post-condition. For now, let us ignore the annotations inside $\{\}$, they will be explained later in section~\ref{chap4:sec:axiomatic}.
\begin{figure}[htbp]
  \centering
  \includegraphics[width=0.8\linewidth]{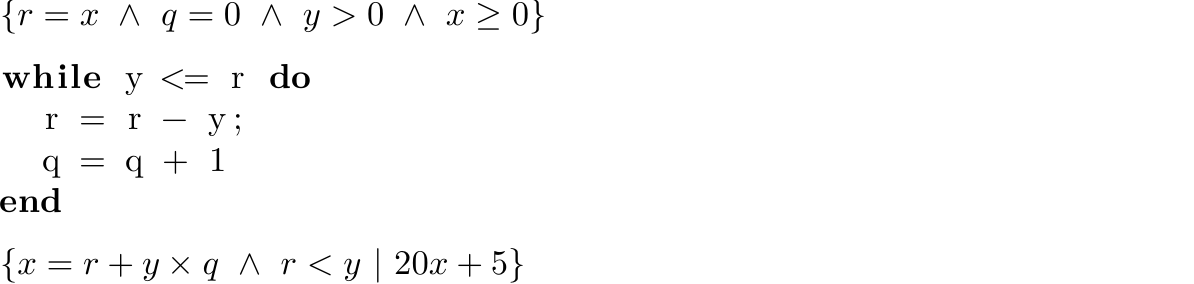}
  \caption{Division algorithm implemented in IMP.}
  \label{fig:division}
\end{figure}
%

\section{Operational Semantics}\label{chap4:sec:formal-operational}
Now that we have defined the syntax of our language we need to set semantic rules that will give the meaning of a program.
In order to achieve this, and since our language has variable declarations, we first need to define the notion of \emph{state}.
A state can be defined as a function that given a variable will return its value.
$$State: var \rightarrow int$$
A variable is defined as either an identifier or a position in an array.
$$var ::= x\ |\ x [n]$$
Thus, writing $\sigma\ x$ will specify the value of variable $x$ in state $\sigma$.

\subsection{Semantic of Expressions}

In order to evaluate arithmetic expressions, we define a semantic function $\mathcal{A}$ which will receive two arguments, an arithmetic expression and a state.
$$\mathcal{A}: aexp \rightarrow state \rightarrow int$$

Writing $\mathcal{A} \llbracket a \rrbracket \sigma$ will return the value of evaluating expression $a$ in state $\sigma$. The function is defined in figure~\ref{fig:arith}.
\begin{figure}[htbp]
  \centering
  \includegraphics[width=0.85\linewidth]{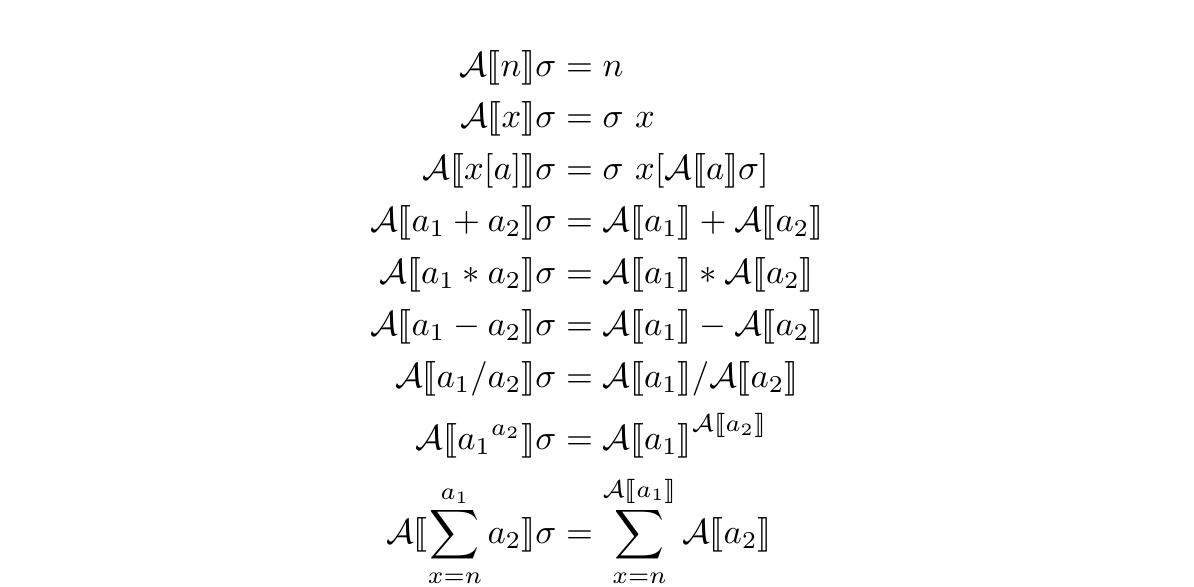}
  \caption{Semantics of arithmetic expressions.}
  \label{fig:arith}
\end{figure}

Similarly, we define a semantic function $\mathcal{B}$ that, given a state, will convert a boolean expression to truth values.
$$\mathcal{B}: bexp \rightarrow state \rightarrow bool$$
In figure~\ref{fig:bool} we define $\mathcal{B}$ using the previous definition of $\mathcal{A}$.
\begin{figure}[htbp]
  \centering
  \includegraphics[width=0.8\linewidth]{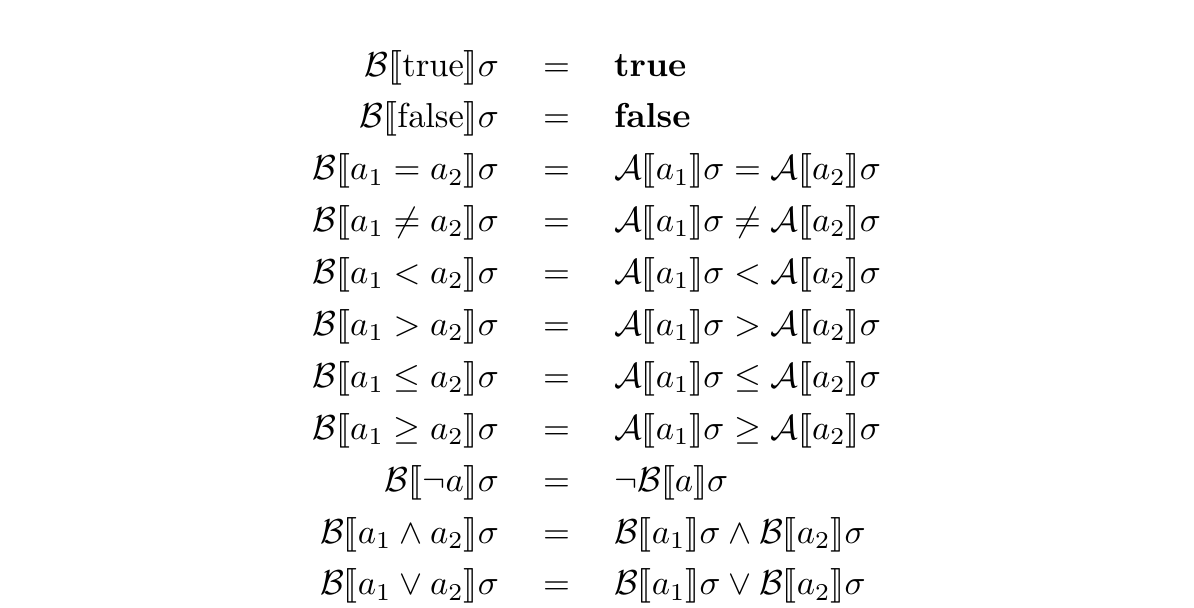}
  \caption{Semantics of boolean expressions.}
  \label{fig:bool}
\end{figure}
%

\subsection{Cost of expressions}

Our semantics will not only evaluate the meaning of a program but also compute the exact cost of executing it. To achieve this, we start by defining semantics for the cost of evaluating arithmetic and boolean expressions, as shown in Figure~\ref{fig:cost}.
\begin{figure}[htbp]
  \centering
  \includegraphics[width=0.8\linewidth]{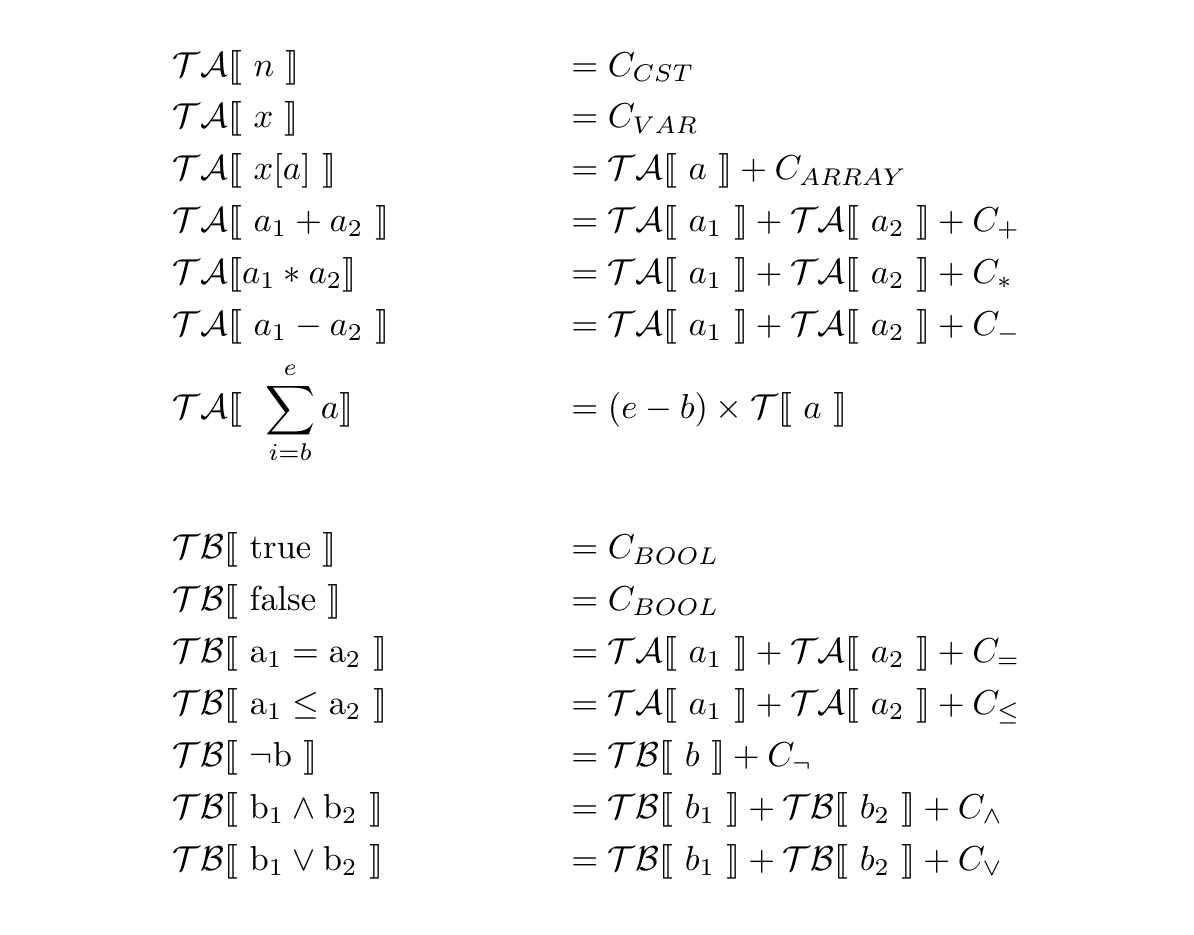}
  \caption{Cost of Arithmetic and Boolean Expressions.}
  \label{fig:cost}
\end{figure}
To evaluate the cost of an expression, we must establish the cost of atomic operations in our language, such as reading from memory or performing basic arithmetic (addition, multiplication, etc) and logic operations (disjunction, negation, etc). 
For example, $C_{CST}$ corresponds to the cost of evaluating a constant, and $C_{VAR}$ is the cost of evaluating a variable. The cost of evaluating a multiplication $\mathcal{T}\mathcal{A}\llbracket a_1 * a_2\rrbracket$ is defined as a sum of the cost of evaluating each of the arithmetic expressions, $a_1$ and $a_2$, plus the cost of the multiplication operation $C_{*}$. 
The cost of evaluating a sum $\sum_{i=b}^{e}a$ is the cost of evaluating $a$ multiplied by the number of times we evaluate it $e - b$.
These rules are simultaneously used by our operational semantics when executing our program, and by our axiomatic semantics, when proving time restrictions statically using our VCG.
For simplicity in the rest of the document, we will consider all the atomic costs as $1$, except in logic definitions and soundness proofs.
\subsection{Free Variables and Substitution}

Before we can define the semantics of a statement we need to first look at two important definitions: Free Variables, and Substitution.

\begin{definition}[Free Variables]
The Free Variables of an arithmetic expression can be defined as the set of variables occurring in an expression that are not bounded by any variable binding operator, such as $\sum$.\\
If we define this as a function $FV: a \rightarrow \{x\}$ we get the definition in figure~\ref{fig:free}.
\begin{figure}[htbp]
  \centering
  \includegraphics[width=0.8\linewidth]{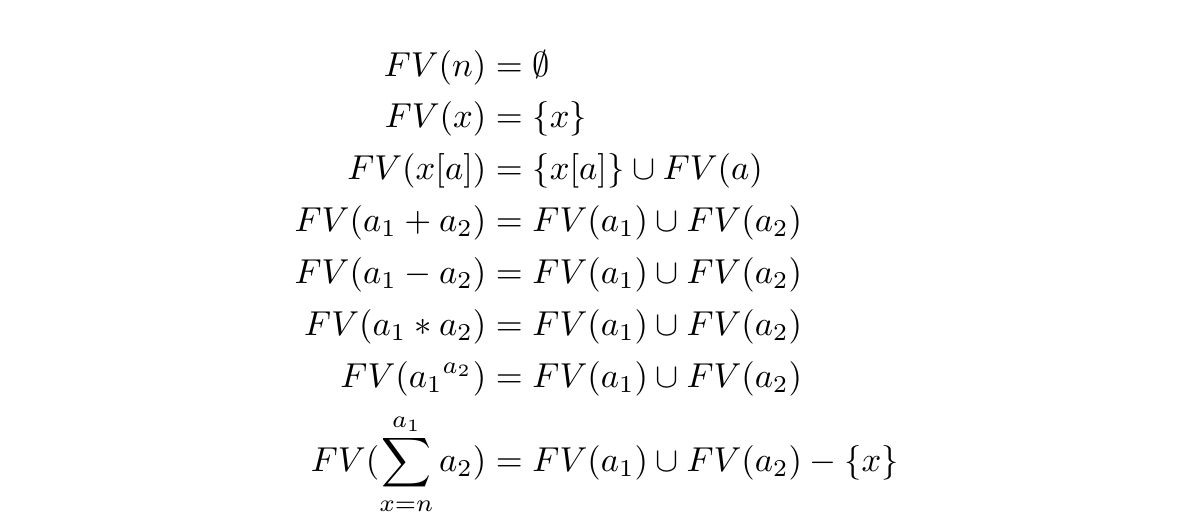}
  \caption{Free Variables of Arithmetic Expressions.}
  \label{fig:free}
\end{figure}
\end{definition}

For example, the free variables of $\sum_{x=0}^{10}x + 2^y + z$ are $\{y,z\}$.

\begin{definition}[Substitution]
A substitution consists of replacing every occurrence of a variable ($x_1$) in an arithmetic expression ($a$) with another arithmetic expression ($a_0$). This is written as $a[a_0/x_1]$ and the substitutions rules are as described in figure~\ref{fig:subst}.
\begin{figure}[htbp]
  \centering
  \includegraphics[width=0.8\linewidth]{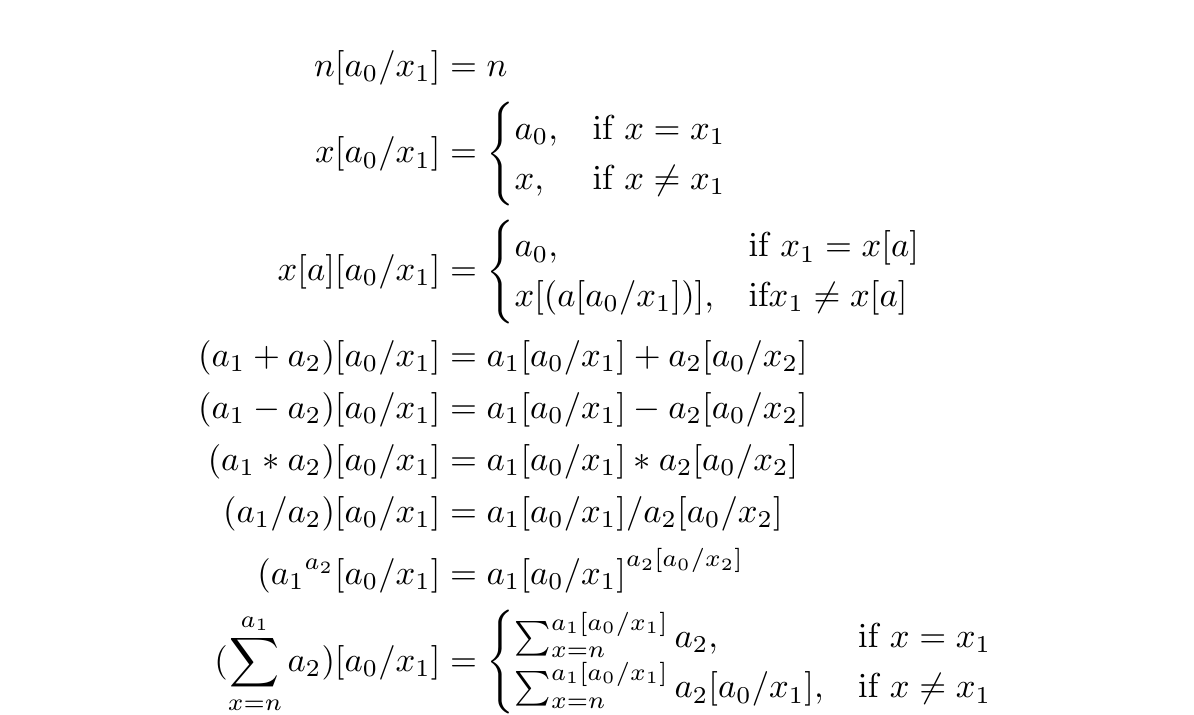}
  \caption{Substitution algorithm for arithmetic expressions.}
  \label{fig:subst}
\end{figure}
\end{definition}

As an example, let us look at the following substitution 
$$(x+4y+3)[z+4/y] = x + 4(z+4) + 3$$

Substitutions might also be applied to states. For example, $\sigma[n/x]$ represents a state that is identical to $\sigma$, with the exception that $x$ takes the value of $n$. Note that, since a state is a mapping from variable to an integer value, $n$ has to always be an integer and never an expression. 

\subsection{Evaluating statements}

In order to prove assertions on the execution time of a program, we need to define a cost-aware semantics. We define a natural operational semantics, where transitions are of the form $$\langle S, \sigma \rangle \rightarrow^t \sigma'$$
meaning that after executing statement $S$ from state $\sigma$ the final state is $\sigma'$ and the execution time was $t$. 

The cost-instrumented operational semantics is defined in Figure~\ref{fig:operational}. 
\begin{figure}[htbp]
  \centering
  \includegraphics[width=0.8\linewidth]{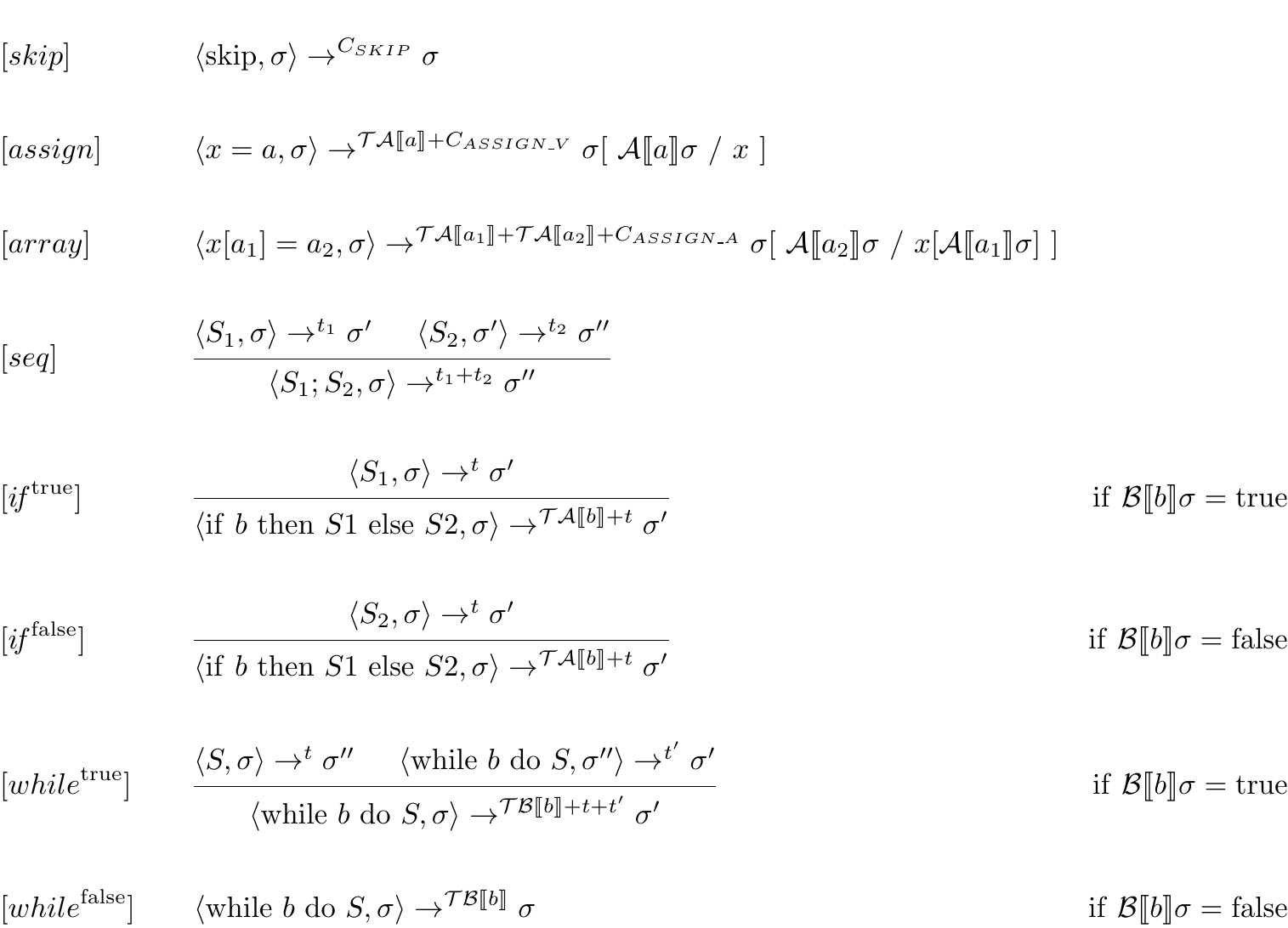}
  \caption{Operational Semantics.}
  \label{fig:operational}
\end{figure}

The skip axiom $[skip]$ says that \textit{skip} does not change the state of the program, and we associate a constant cost for its execution of $C_{SKIP}$.

The assignment axiom $[assign]$ says that executing the assignment $x = a$ in state $\sigma$ will lead to a state $\sigma[\mathcal{A}\llbracket a \rrbracket \sigma / x]$, which means state $\sigma$ where x takes the value of $\mathcal{A}\llbracket a \rrbracket \sigma$. The cost of this expression is defined as the cost of evaluating $a$, $\mathcal{T}\mathcal{A} \llbracket a \rrbracket$, plus the constant cost of an assignment, $C_{ASSIGN\_V}$.

Similarly, the array assignment axiom $[array]$ says that executing $x[a_1] = a_2$ from state $\sigma$ will lead to state $\sigma[\mathcal{A} \llbracket a_2 \rrbracket \sigma / x[\mathcal{A} \llbracket a_2 \rrbracket s\sigma]$, which is a similar state to $\sigma$, except $x[\mathcal{A} \llbracket a_1 \rrbracket \sigma]$ takes the value of $\mathcal{A} \llbracket a_2 \rrbracket \sigma$. This execution cost will be the cost of evaluating $a_1$, $\mathcal{T}\mathcal{A} \llbracket a_1 \rrbracket$, plus the cost of evaluating $a_2$, $\mathcal{T}\mathcal{A} \llbracket a_2 \rrbracket$, plus the constant cost of assigning a value to a position in an array, $C_{ASSIGN\_A}$.

The sequence rule $[seq]$ says that if we want to execute a sequence $S_1;S_2$ from state $\sigma$ we will first execute statement $S_1$ from state $\sigma$, this execution will lead to a certain state $\sigma'$ in $t_1$ time. If we execute $S_2$ from this state $\sigma'$ we will reach a final state $\sigma''$ in $t_2$ time. Therefore executing $S_1;S_2$ from state $\sigma$ will lead to state $\sigma''$ in $t_1+t_2$ time.

We have two conditional rules, $[if^{true}]$ and $[if^{false}]$. In order to decide which of the rules to apply, we must first evaluate $\mathcal{B}\llbracket b \rrbracket \sigma$. If this evaluates to true we apply rule $[if^{true}]$, which means we simply execute statement $S_1$, otherwise we apple rule $[if^{false}]$, which means we execute statement $S_2$. For both rules, the cost of executing the if statement is the cost of executing $S_1$ when true or $S_2$ when false, plus the cost of evaluating $b$, $\mathcal{T}\mathcal{B}\llbracket b \rrbracket$.

We have one rule and one axiom for while, $[while^{true}]$ and $[while^{false}]$. 

If $\mathcal{B} \llbracket b \rrbracket \sigma$ is false, we apply the axiom $[while^{false}]$, that says we will remain in the same state $\sigma$ and the cost is simply the cost of the evaluation of b, $\mathcal{T}\mathcal{B} \llbracket b \rrbracket$. 

If $\mathcal{B} \llbracket b \rrbracket \sigma$  is true, we apply rule $[while^{true}]$, which means we will execute the loop body, $S$, once from state $\sigma$ and this will lead to a state $\sigma''$. Finally, we execute the while loop again, but this time from state $\sigma''$. The cost of the while loop, in this case, is the cost of evaluating b, plus the cost of executing the body, plus the cost of executing the while loop from state $\sigma''$.

\subsection*{Example}

Let us consider the following program that swaps the values of x and y:
\begin{lstlisting}[language=ml]
z = x; x = y; y = z
\end{lstlisting}
Let the initial state $\sigma_0$ be a state such that $\sigma_0\ x = 3$, $\sigma_0\ y = 10$ and $\sigma_0\ z = 0$.

Then the derivation tree of this program will look like

\begin{prooftree}
    \infer0{\langle z=x, \sigma_0 \rangle \rightarrow^{t_1} \sigma_1}
    \infer0{\langle x=y, \sigma_1 \rangle \rightarrow^{t_2} \sigma_2}
    \infer2{\langle z=x; x=y, \sigma \rangle \rightarrow^{t_1 + t_2} \sigma_2}
    \infer0{\langle y=z , \sigma_2 \rangle \rightarrow^{t_3} \sigma_3}
    \infer2{\langle z=x; x=y; y=z , \sigma_0 \rangle \rightarrow^{t_1 + t_2 + t_3} \sigma_3}  
\end{prooftree}

Where $\sigma_1 = \sigma_0[3 / z]$, $\sigma_2 = \sigma_1[10/x]$ and $\sigma_3 = \sigma_2[3/y]$. 

From the assign rule we get that $t_1 = \mathcal{TA} \llbracket x \rrbracket + C_{ASSIGN\_V} = C_{VAR} + C_{ASSIGN\_V}$, $t_2$ and $t_3$ will be the same. Therefore $t_1 + t_2 + t_3 = 3 \times (C_{VAR} + C_{ASSIGN\_V})$.

\section{Axiomatic Semantics}\label{chap4:sec:axiomatic}
We will now define a logic with triples in the form $\{ P \} S \{Q | t\}$. This triple can be read as executing $S$ from a state $\sigma$ that validates the precondition $P$ leads to a state that validates postcondition $Q$, and this execution costs at most $t$ to complete. We will call these triples {\em partial correctness assertions}.

Before defining our assertion language, we must distinguish between program variables and logic variables. 
Let us imagine the following triple $\{x=n\} y=x+1\{y>n\}$. In this case, $x$ and $y$ are program variables since they are both present in the statement inside the triple. In our pre and postconditions, we reference variable $n$, which does not appear on the program. This is what we call a \textit{logic variable}.

\paragraph{Assertions}
Our language supports annotations with preconditions and postconditions in order to prove correctness, termination, and time restrictions on our program. Our assertion language will be an extension of the boolean expressions extended with quantifiers over integer variables and implications. 
\begin{align*}
    P  ::&=  true\ |\ false\ \\
    &|\ a_1 = a_2\ 
    |\ a_1 \ne a_2\ 
    |\ a_1 < a_2\ 
    |\ a_1 > a_2\ 
    |\ a_2 \le a_2\ 
    |\ a_1 \ge a_2\\
    &|\ \neg P\ 
    |\ P \land Q\ 
    |\ P \lor Q\ 
    |\ P \rightarrow Q\\
    &|\ \exists x.\ P\ 
    |\ \forall x.\ P
\end{align*}

Note that $a_1$ and $a_2$ are arithmetic expressions extended with logic variables.
As assertions may include logic variables, the semantics require an interpretation function to provide the value of a variable. Given an interpretation, it is convenient to define the states which satisfy an assertion. We will use the notation  $\sigma \models^I P$ to denote that state $\sigma$ satisfies $P$ in interpretation $I$, or equivalently that assertion $P$ is true at state $\sigma$, in interpretation $I$. The definition of $\sigma \models^I P$ is the usual for a first-order language.

Now, note that we are not interested in the particular values associated with variables in an interpretation $I$. We are interested in whether or not an assertion is true at all states for all interpretations. This motivates the following definition. 
\begin{definition}[Validity]
$\models \{ P \} S \{ Q | t \}$
if and only if, for every state $\sigma$ and interpretation $I$ such that $\sigma \models^I P$ and
$\langle S, \sigma \rangle \rightarrow^{t'} \sigma'$ we have that $\sigma' \models^I Q$ and $\mathcal{A}\llbracket t \rrbracket\sigma \ge t'.$
\end{definition}
If $\models \{ P \} S \{ Q | t \}$ we say that the partial correctness assertion $\{ P \} S \{ Q | t \}$ is \emph{valid}.

We now define a set of proof rules that generate valid partial correctness assertions. The rules of our logic are given in figure~\ref{fig:hoare}.
\begin{figure}[htbp]
  \centering
  \includegraphics[width=0.85\linewidth]{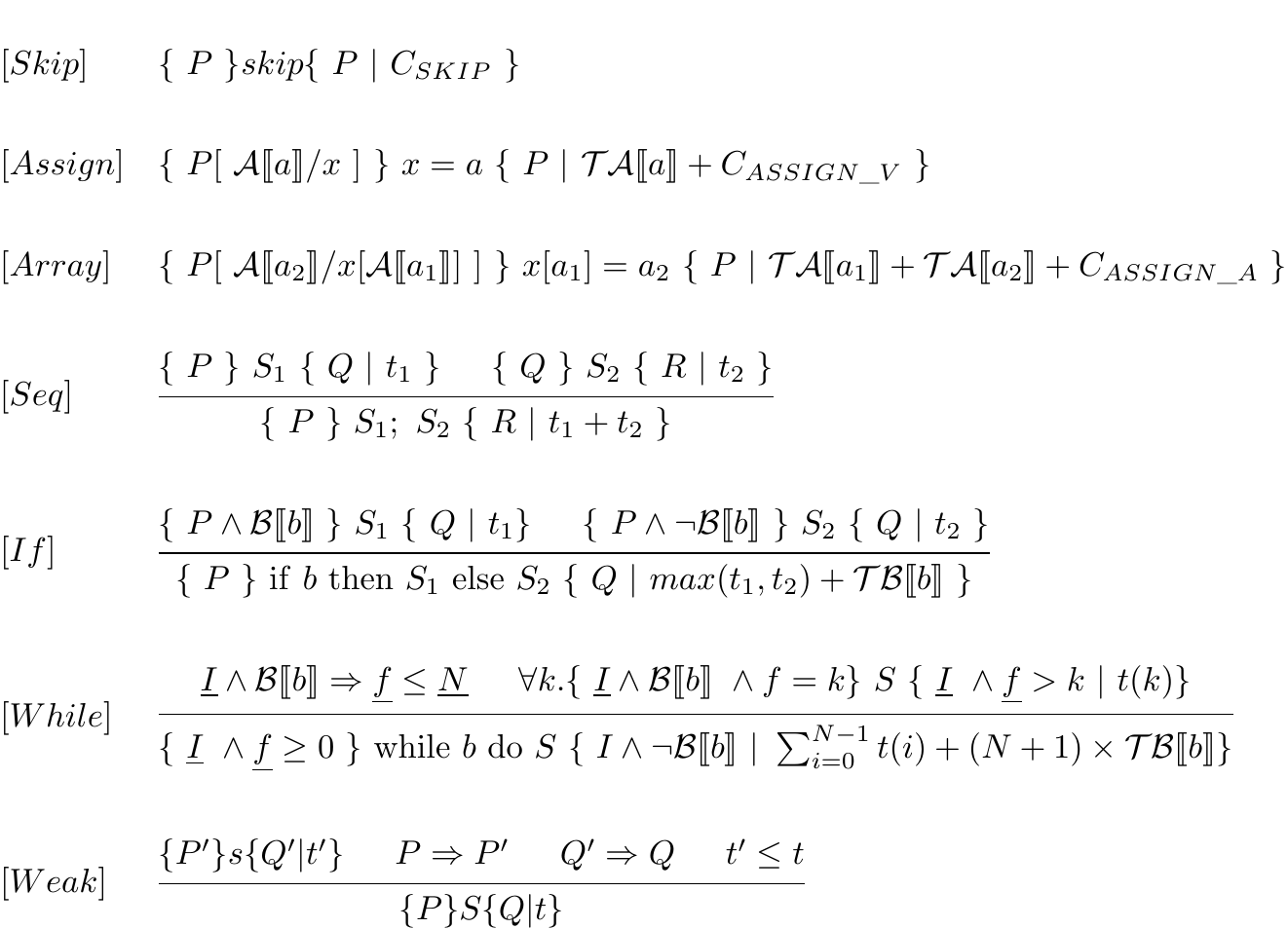}
  \caption{Proof Rules for the IMP Language.}
  \label{fig:hoare}
\end{figure}
%

The \textit{skip axiom} says that if P is true before executing \textit{skip}, then it is also true after its execution. The upper bound for this execution is the constant $C_{SKIP}$.

The \textit{assign axiom} says that $P$ will be true after executing $x = a$ if $P[\mathcal{A}\llbracket a \rrbracket /x]$ is true before its execution. The upper bound of this execution is defined as the sum of evaluating $a$ and assigning a simple variable $\mathcal{T}\mathcal{A}\llbracket a \rrbracket + C_{ASSIGN\_V}$.

The \textit{array axiom} is very similar to the assignment axiom, except in the upper bound we also need to consider the time to evaluate $a_1$.

The \textit{seq rule} says if $P$ is true before executing $S_1;S_2$ then $R$ will be true after the execution, as long as we can prove that 
\begin{itemize}
    \item If $P$ is true before executing $S_1$ then Q is true after and $t_1$ is an upper bound on this execution
    \item If $Q$ is true before executing $S_2$ then $R$ is true after and $t_2$ is an upper bound on this execution
\end{itemize}
The upper bound for the sequence is the sum of both upper bounds, $t_1 + t_2$.

The \textit{if rule} says if $P$ is true before executing if $b$ then $S_1$ else $S_2$, then Q will be true after the execution, as long as we can prove that
\begin{itemize}
    \item If $P$ and $\mathcal{B}\llbracket b \rrbracket$ are both true before executing $S_1$, then Q will be true after executing it. This execution is upper bounded by $t_1$
    \item If $P$ and $\neg \mathcal{B} \llbracket b \rrbracket$ are both true before executing $S_2$, then $Q$ will be true after executing it. This execution is upper bounded by $t_2$
\end{itemize}

In the \textit{while rule}, $f$, $I$, and $N$ are values provided by an oracle: in an interactive proof system, these can be provided by the user, and in a non-interactive setting, they can be annotated into the program. $I$ is the loop invariant, which must remain true before and after every iteration of the loop. $f$ is a termination function that must start as a positive value, increase with every iteration, but remain smaller than $N$. $N$ is therefore the maximum number of iterations. $t(k)$ is a function describing the cost of the body of the while at iteration $k$.
If the {\em while} loop runs at most $N$ times and, for each $k$ iteration, the cost of the loop body is given by $t(k)$ we have the following upper bound for the while statement: the sum of the cost of all the iterations $\sum_{i=0}^{N-1}t(i)$; plus the sum of evaluating the loop condition, $b$, each time we enter the while body (at most N), plus one evaluation of $b$ when the condition fails, and the loop terminates. This leads to the term $\sum_{i=0}^{N-1}t(i) + (N+1) \times \mathcal{T}\mathcal{B}\llbracket b \rrbracket$ used in the rule for the {\em while} loop.

Besides each rule and axiom for every statement in our language, we need an extra rule that we call \textit{weak rule}.
This rule says
\begin{itemize}
    \item If a weaker precondition is sufficient, then so is a stronger one. If we know $\models \{P\}S\{Q | t\}$ and $P' \rightarrow P$, then $\models \{P'\}S\{Q | t\}$,
    \item If a stronger postcondition is provable, then so is a weaker one. If we know $\models \{P\}S\{Q | t\}$ and $Q \rightarrow Q'$, then $\models \{P\}S\{Q' | t\}$
    \item An upper bound can always be replaced with a bigger one. If we know $\models \{P\}S\{Q|t\}$ and $t' > t$, then $\models \{P\}S\{Q|t'\}$
\end{itemize}

Proof rules should preserve validity in the sense that if the premise is valid, then so is the conclusion. 
When this holds for every rule, we say that the proof system is {\em sound}. For the proposed Hoare logic, it is generally easy to show by induction that every assertion $\{ P \} S \{ Q | t \}$ which is the conclusion of a derivation in the proof system, is a valid assertion. 

The proof of soundness depends on an important property of substitution.

\begin{lemma}
\label{lemma:subst}
Let $P$ be an assertion, $x$ a variable, and $a$ an arithmetic expression. Then for every state $\sigma$
 $$ \sigma \models P[a / x] \quad \textit{iff} \quad \ \sigma[\mathcal{A}\llbracket a \rrbracket\sigma / x] \models P$$
\end{lemma}
\begin{proof}
The proof follows by structural induction on $P$.
\end{proof}

\begin{theorem}[Soundness]
\label{theorem:soundness}
Let $\{ P \} S \{ Q | t \}$ be a partial correctness assertion. Then
$$\vdash \{ P \} S \{ Q | t \} \quad \rightarrow \quad \models \{ P \} S \{ Q | t \}$$
\end{theorem}
\begin{proof}
The proof follows by induction on the length of the derivation of $\{ P \} S \{ Q | t\}$. 

{\em Case Skip:}
Consider a state $\sigma$, where $\sigma \models P$. According to our semantic, after executing \textit{skip} we are still in state $\sigma$ and the statement \textit{skip} takes $C_{SKIP}$ to execute, $\langle \sigma, \mathrm{skip} \rangle \rightarrow^{C_{SKIP}} \sigma$.
Therefore, the rule for \textit{skip} is sound.
$$\models \{P\} skip \{ P | C_{SKIP}\}$$

{\em Case Assign:}
Consider a state $\sigma$ where $\sigma \models P[\mathcal{A}\llbracket a \rrbracket / x]$. 
We have $ \langle \sigma, x = a \rangle \rightarrow^{\mathcal{T}\mathcal{A} \llbracket a \rrbracket + C_{ASSIGN\_V}} \sigma[\mathcal{A}\llbracket a \rrbracket\sigma\ /\ x]$, that is, after the statement is executed we are in state $\sigma[\mathcal{A}\llbracket a \rrbracket\sigma\ /\ x]$ and it costs $\mathcal{T}\mathcal{A} \llbracket a \rrbracket + C_{ASSIGN\_V}$. By the substitution lemma~\ref{lemma:subst} we have that $\sigma[\mathcal{A}\llbracket a \rrbracket\sigma\ /\ x] \models P$.
Therefore, the rule for assignment is sound.
$$ \models \{\ P[\ \mathcal{A} \llbracket a \rrbracket / x\ ]\ \}\ x = a\ \{\ P\ |\ \mathcal{T}\mathcal{A}\llbracket a \rrbracket + C_{ASSIGN\_V}\ \}$$

{\em Case Array:}
Consider a state $\sigma$ where $\sigma \models P[\mathcal{A}\llbracket a_2 \rrbracket / x[\mathcal{A} \llbracket a_1 \rrbracket]]$. 
We have 
$$ \langle \sigma, x[a_1] = a_2 \rangle \rightarrow^{\mathcal{T}\mathcal{A} \llbracket a \rrbracket + C_{ASSIGN\_A}} \sigma[\mathcal{A}\llbracket a_2 \rrbracket\sigma/x[\mathcal{A} \llbracket a_1 \rrbracket\sigma]]$$
That is, after the statement is executed we are in state 
$$\sigma[\mathcal{A}\llbracket a_2 \rrbracket\sigma/x[\mathcal{A} \llbracket a_1 \rrbracket\sigma]]$$ 
and the cost is $\mathcal{T}\mathcal{A} \llbracket a \rrbracket + C_{ASSIGN\_A}$.
By lemma~\ref{lemma:subst} we have that  $\sigma[\mathcal{A}\llbracket a_2 \rrbracket\sigma/x[\mathcal{A} \llbracket a_1 \rrbracket\sigma]] \models P$.
Therefore, the rule for assignment to an array is sound.

$$\models \{\ P[\ \mathcal{A} \llbracket a_2 \rrbracket / x[\mathcal{A} \llbracket a_1 \rrbracket]\ ]\ \}\ x[a_1] = a_2\ \{\ P\ |\ \mathcal{T}\mathcal{A}\llbracket a_1 \rrbracket + \mathcal{T}\mathcal{A}\llbracket a_2 \rrbracket + C_{ASSIGN\_A}\ \}$$

{\em Case Sequence:}
Assume $\models \{P\}S_1\{Q|t_1\}$ and $\models \{Q\}S_2\{R|t_2\}$. Lets consider a state $\sigma$ that validates P, $\sigma \models P$.
From our semantics we have that, there is a state $\sigma_1$ such that $\langle S_1, \sigma \rangle \rightarrow^{t} \sigma_1$. Since we have that $\models \{P\}S_1\{Q|t_1\}$, then $\sigma_1 \models Q$, and $\mathcal{A} \llbracket t_1 \rrbracket \sigma \ge t$. 
We also have that $\langle S_2, \sigma_1 \rangle \rightarrow^{t'} \sigma_2$, and since we assume $\models \{Q\}S_2\{R|t_2\}$, then $\sigma_2 \models R$, and $\mathcal{A} \llbracket t_2 \rrbracket \sigma_1 \ge t'$. 
Since $t_1$ and $t_2$ are both constant values, $\mathcal{A} \llbracket t_2 \rrbracket \sigma_1 = \mathcal{A} \llbracket t_2 \rrbracket \sigma$, which means $\mathcal{A} \llbracket t_1 + t_2 \rrbracket \sigma \ge t + t'$.
Therefore the proof for \textit{sequence} is sound.
$$\models \{P\}S_1;S_2\{Q| t_1 + t_2 \}$$

{\em Case If:}
Assume $\models \{ P \land b \} S_1 \{ Q | t_1\}$ and $\models \{P \land \neg b\} S_2 \{Q | t_2\}$.
Suppose $\sigma \models P$. 

If $\sigma \models b$, then $\sigma \models P \land b$ so, assuming $\langle S_1, \sigma \rangle \rightarrow^{t} \sigma_1$, we have that $\sigma_1 \models Q$, and $\mathcal{A} \llbracket t_1 \rrbracket \sigma \ge t$.

If $\sigma \models \neg b$, then $\sigma \models P \land \neg b$ so, assuming $\langle S_2, \sigma \rangle \rightarrow^{t'} \sigma_2$, we have that $\sigma_2 \models Q$, and $\mathcal{A} \llbracket t_2 \rrbracket \sigma \ge t'$.

Since $\mathcal{A} \llbracket \mathrm{max}(t_1,t_2) + \mathcal{T}\mathcal{B}\llbracket b \rrbracket \rrbracket \sigma \ge \mathcal{T}\mathcal{B}\llbracket b \rrbracket + t$ and $\mathcal{A} ~\llbracket \mathrm{max}(t_1,t_2) + \mathcal{T}\mathcal{B}\llbracket b \rrbracket \rrbracket \sigma \ge \mathcal{T}\mathcal{B}\llbracket b \rrbracket + t'$. The rule for \textit{if} is sound.

$$\models \{\ P\ \}\ \mathrm{if}\ b\ \mathrm{then}\ S_1\ \mathrm{else}\ S_2\ \{\ Q\ |\ max(t_1,t_2) + \mathcal{T} \mathcal{B} \llbracket b \rrbracket\ \}$$
 
{\em Case While:}
Assume $I \land \mathcal{B} \llbracket b \rrbracket \Rightarrow f \le N$ and $\models \{\ I \land \mathcal{B} \llbracket b \rrbracket\ \land f=k\}\ S\ \{\ I\ \land f > k\ |\ t(k)\}$. 
Considering a state $\sigma$ such that $\sigma \models I \land f \ge 0$ and $\langle \mathrm{while}\ b\ \mathrm{do}\ S, \sigma\rangle \rightarrow^{t} \sigma_1$.

If $\sigma \models \neg \mathcal{B} \llbracket b \rrbracket$ then $\sigma_1 = \sigma$, therefore $\sigma_1 \models \neg \mathcal{B} \llbracket b \rrbracket \land I$ and the cost is $t = \mathcal{T}\mathcal{B}\llbracket b \rrbracket$.

If $\sigma \models \mathcal{B} \llbracket b \rrbracket$, we have that $\sigma \models \mathcal{B} \llbracket b \rrbracket \land I$.
Considering a state $\sigma_2$ such that $\langle S, \sigma \rangle \rightarrow^{t'} \sigma_2$ and $\langle \mathrm{while}\ e\ \mathrm{do}\ S, \sigma_2\rangle \rightarrow^{t''} \sigma_1$. Applying our initial assumption we get $\sigma_2 \models I \land f > k$, and $\mathcal{A} \llbracket t(k) \rrbracket \sigma \ge t'$.
Finally by applying the induction hypothesis we have that $\sigma_1 \models \neg \mathcal{B} \llbracket b \rrbracket \land I$.

Given the function f provided by the user and taking into account our assumptions, we know that the program will eventually stop and at most, it will iterate $N+1$ times.
For each iteration the cost is $\mathcal{T} \mathcal{B} \llbracket b \rrbracket + t'$ and $\mathcal{T} \mathcal{B} \llbracket b \rrbracket$ for the last time, when $b$ is false and the while terminates.
Therefore the cost of this program is always smaller or equal to 
$$\mathcal{T} \mathcal{B} \llbracket b \rrbracket + \sum_{i=0}^{N} (\mathcal{T} \mathcal{B} \llbracket b \rrbracket + t')$$
Since $t(k) \ge t'$, our upper bound is
$$\sum_{i=0}^{N+1} \mathcal{T} \mathcal{B} \llbracket b \rrbracket + \sum_{i=0}^{N} t(k)$$

Therefore the rule for \textit{while} is sound.

$$\models \{\ I\ \land f \ge 0\ \}\ \mathrm{while}\ b\ \mathrm{do}\ S\ \{\ I \land \neg \mathcal{B} \llbracket b \rrbracket\ |\ \sum_{i=0}^{N} t(i) + \sum_{i=0}^{N+1} \mathcal{T}\mathcal{B}\llbracket b \rrbracket\}$$
\end{proof} 

\subsection*{Example: Division Algorithm}

To illustrate our logic, let us apply our rules to the division algorithm, as presented in figure~\ref{fig:division}.

Since our example contains a while loop we will have to define, the invariant $I$ as $x = q \times y + r \land y \ge 0 \land r \ge 0$, the variant $f$ as $x-r$, the maximum number of iterations $N$ as $x$, and the function of cost $t(k)$ as $fun\ k \rightarrow 10$.

Let us consider $S_w \equiv r=r-y;q=q+1$.

By the \textit{assign rule} we know
\begin{equation}
    \vdash \{ (I \land f > k)[q+1/q] \}q=q+1\{ I \land f > k | \mathcal{TA} \llbracket q+1 \rrbracket + 1 \} \label{chap4:eq-1}
\end{equation}

Where
$$\mathcal{TA} \llbracket q+1 \rrbracket = \mathcal{TA} \llbracket q \rrbracket + \mathcal{TA} \llbracket 1 \rrbracket + 1 = 3$$
$$(I \land f > k)[q+1/q] \equiv x = (q+1)\times y+r \land y \ge 0 \land r \ge 0 \land x-r>k$$

Again by the \textit{assign rule} we know
\begin{equation}
  \vdash \{ (I \land f > k)[q+1/q][r-y/r] \}r=r-y\{ (I \land f > k)[q+1/q] | \mathcal{TA} \llbracket r-y \rrbracket + 1 \}  \label{chap4:eq-2}
\end{equation}

Where
\begin{align*}
&\mathcal{TA} \llbracket r-y \rrbracket = \mathcal{TA} \llbracket r \rrbracket + \mathcal{TA} \llbracket y \rrbracket + 1 = 3\\
&(I \land f > k)[q+1/q][r-y/r] \equiv x = (q+1) \times y+r-y \land y \ge 0 \land r-y \ge 0 \land x-r + y>k 
\end{align*}

Since \ref{chap4:eq-1}, and \ref{chap4:eq-2}, by the \textit{seq rule}, we have
\begin{equation}
\vdash \{ (I \land f > k)[q+1/q][r-y/r] \}S_w\{ I \land f > k | 8 \}  
\end{equation}
\begin{align*}
I \land y \le r \land f=k \quad \equiv& \quad x = q \times y + r \land y \ge 0 \land r \ge 0 \land y \ge r \land x-r=k\\
&\rightarrow \quad x = (q+1) \times y + r - y \land y \ge 0 \land r \ge y \land x-r+y \ge k\\
&\equiv \quad (I \land f>k)[q+1/q][r-y/r]
\end{align*}

Given this result, and since $10 \ge 8$, by the \textit{weak rule} we get
\begin{equation}
    \vdash \{I \land y \le r \land f = k \} r=r-y;q=q+1 \{I \land f>k | 10 \}
\end{equation}

It also apparent that 
$$I \land y \le r \rightarrow r \ge 0 \rightarrow x- r \le x$$

We are now ready to apply the \textit{while rule}, and we get
$$\vdash \{I \land f \ge 0 \} S \{ I \land \neg (y \le r) | \sum_{i=0}^{x-1}10 + (x+1) \times \mathcal{TB} \llbracket y \le r \rrbracket \}$$

Where
$$\sum_{i=0}^{x-1}10 + \mathcal{TB} \llbracket y \le r \rrbracket = 10 \times x + (x+1) \times 3 = 13x + 3$$

Since $P \rightarrow I \land f \ge 0$, $I \land \neg (y \le r)$, and $20x+5 \ge 13x +3$, by the \textit{weak rule} we get
$$\vdash \{P\} S \{Q | T \}$$

\section{Verification Conditions Generation}\label{chap4:sec:vcg}
To implement a verification system based on our logic, we define a VCG based on the weakest-precondition algorithm.

 Given a Hoare triple $\{P\} S \{ Q | t\}$, we start by identifying the weakest precondition of $S$ given $Q$ as the postcondition. In figure~\ref{fig:wp} we define the wpc function which receives a statement and a postcondition and returns a tuple $(wp, t_S)$, where $wp$ is the weakest precondition of the program and $t_S$ is an upper bound on the program's cost. The weakest precondition is calculated  by a standard algorithm  such as the one presented in~\cite{almeida2011}. Let us focus on the second value of our tuple, which estimates an upper bound for the program. This upper bound will be equivalent to the ones presented in figure~\ref{fig:arith}. The upper-bound for \textit{skip}, \textit{assignment}, and \textit{array} are very straightforward since they are the same as the exact cost calculated by our operational semantics in section~\ref{chap4:sec:formal-operational}. To calculate the upper bound of a sequence $S_1; S_2$ we need to both calculate the upper bound of $S1$ given by $t_1$ and the upper bound of $S_2$, given by $t_2$. The upper bound for the sequence is then defined as the sum of both upper bounds $t_1 + t_2$.
 In the case of \textit{if}, we calculate the upper bound of each conditional statement $S1$ and $S_2$. The upper bound for if is defined as the max between both upper bounds $max(t_1, t_2)$ plus the cost of evaluating $b$, $\mathcal{TB}\llbracket b \rrbracket$.
 Finally, looking at the while, the upper bound is defined as the sum of evaluating the while condition N+1 times and the sum of the cost of each execution of the \textit{while} body, given by t(k).
\begin{figure}[htbp]
  \centering
  \includegraphics[width=0.8\linewidth]{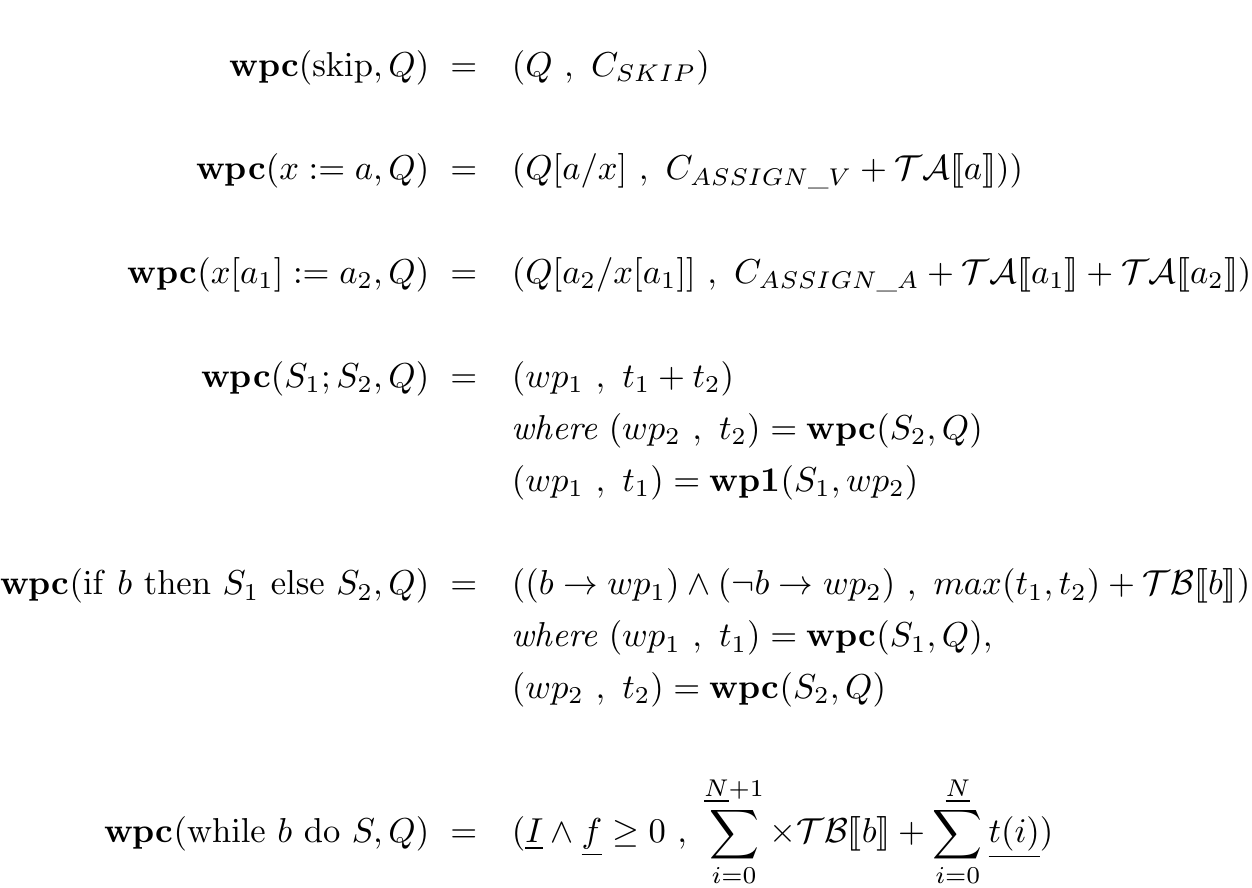}
  \caption{ Weakest Precondition Algorithm.}
  \label{fig:wp}
\end{figure}

The program is correct if $P$ implies the weakest precondition. If $t \ge t_S$, we also guarantee that $t$ is, in fact, an upper bound on the execution of $S$.
Even though these conditions are enough to prove both correction and cost upper-bound of our program, we still need extra VCs to handle while loops. 
The VC function in figure~\ref{fig:vc} receives a program and a postcondition. It returns a set of purely mathematical statements (the verification conditions) needed to handle loops and guarantee termination. Let us take a more in-depth look to the \textit{while} case, $VC($ while $b$ do $S, Q)$. To prove the while statement, we need the oracle to provide some extra information. All the values provided need to be demonstrated. Firstly, as seen previously, the weakest precondition of a while is the invariant being true and the variant being positive. The loop invariant is assured by the following VC
$$\forall k. I \rightarrow wp_S$$
We also need to guarantee the termination of the loop. For this we prove that the variant ($f$) always increases $\forall k, I \land b \land f=k \rightarrow wp_S$ and that the variant is always limited by $N$, $I \land \mathcal{B}\llbracket b \rrbracket \rightarrow f \le N$. Lastly, we prove the postcondition $I \land \neg b \rightarrow Q$, and call $VC$ recursively for the body.
The $VC$ function applied to \textit{sequence} and \textit{if}, is just a recursive call of the function to their sub-statements.
Finally for \textit{skip}, \textit{assign}, and \textit{array} there are no extra VC.
The function $VCG$ is the glue that puts all these VCs together. The first condition $P \rightarrow wp$ implies the correctness of our program. Secondly, we call $VC(S, Q)$ to potentially deal with loop invariants, termination, and cost. The last condition states that $t \ge t_S$ proves if the user-provided upper bound is indeed valid.
These verification conditions are then passed to an interactive prover, such as EasyCrypt \cite{barbosa21}, which attempts to prove them automatically. If it fails, some advice is needed from the user.
\begin{figure}[htbp]
  \centering
  \includegraphics[width=0.8\linewidth]{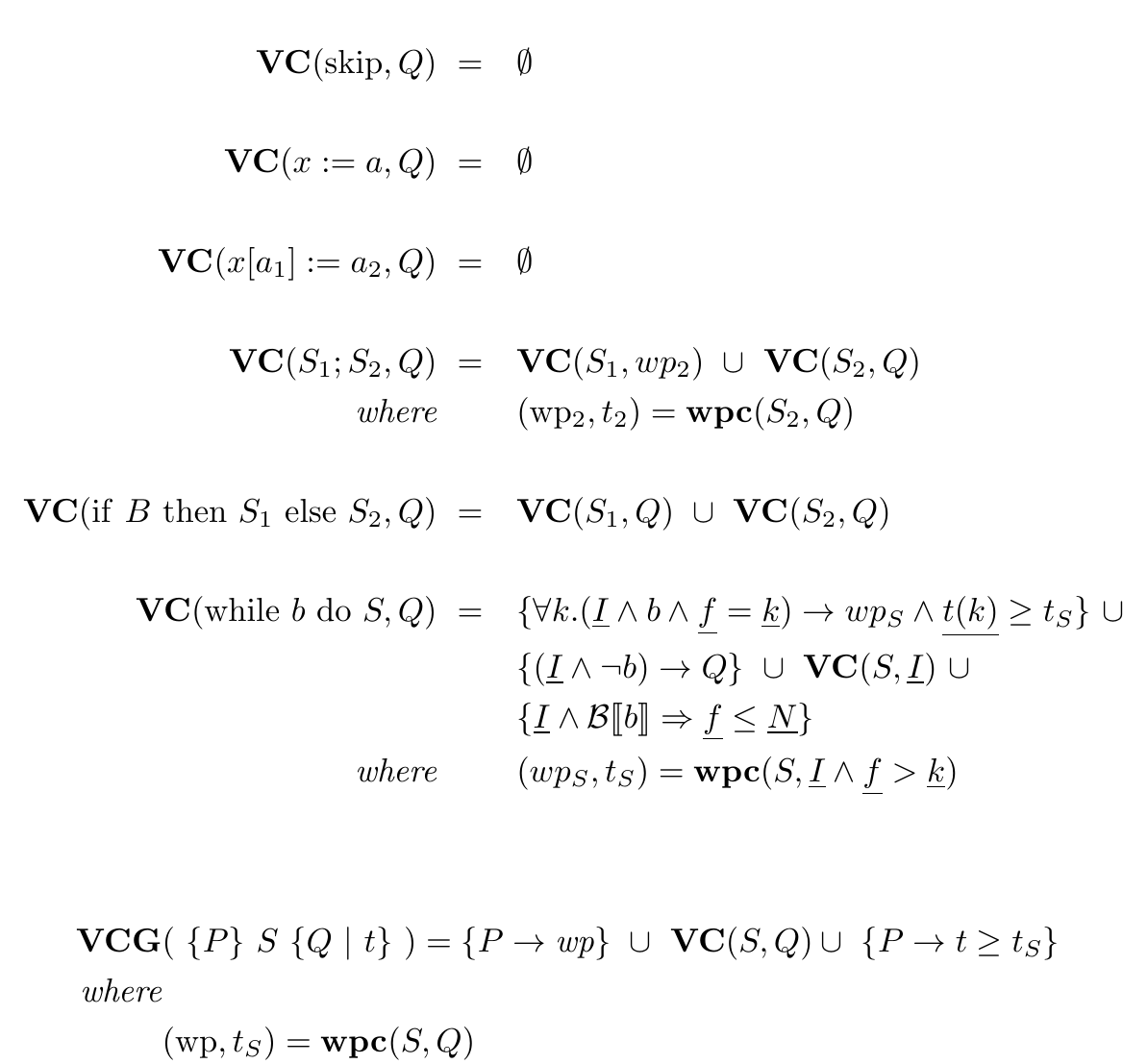}
  \caption{VC and VCG functions.}
  \label{fig:vc}
\end{figure} 

We need to ensure that this algorithm is actually sound with our Hoare Logic. For this, we need to prove theorem~\ref{theor:soundness-vcg} that states that the VCG algorithm is sound if the VC generated implies of the Hoare triple we wish to prove, $\models VCG(\{P\}Q\{R\})\ \rightarrow\ \vdash \{P\}Q\{R\}$. This assures that by proving our VCs, we are actually proving our triple. If we also want to show that if a triple is valid, then we can also validate every VC generated by the VCG algorithm, then we are looking at the completeness theorem~\ref{theor:completeness-vcg}, $\vdash \{P\}Q\{R\}\ \rightarrow\ \models VCG(\{P\}Q\{R\})$. 
To prove both soundness and completeness, we then need to prove
$$\models VCG(\{P\}Q\{R\})\ \textit{iff}\ \vdash \{P\}Q\{R\}$$
\begin{proof}
We prove $\Rightarrow$ by induction on the structure of Q and $\Leftarrow$ by induction in the derivation of $\vdash \{P\}Q\{R\}$.

\textit{Case Skip:}
$$VCG(\{P\} skip \{Q|T\}) = \{P \rightarrow Q \} \cup \{T \ge C_{SKIP}\}$$

By the \textit{skip axio}m, we know $\vdash \{Q\} skip \{Q|C_{SKIP}\}$. \\
By the weak rule, knowing $\vdash \{Q\} skip \{Q|C_{SKIP}\}$, $P \rightarrow Q $, and $T \ge C_{SKIP}$, we have
$$\vdash \{P\} skip \{Q | T \}$$

\textit{Case Assign:}

$$VCG(\{P\} x=a \{Q|T\}) = \{P \rightarrow Q[\mathcal{A} \llbracket a \rrbracket / x] \} \cup \{T \ge C_{ASSIGN\_V}\}$$

By the \textit{assign axiom} $\vdash \{Q[\mathcal{A} \llbracket a \rrbracket / x]\} x=a \{Q|C_{ASSIGN\_V}\}$.
By the weak rule, knowing $\vdash \{Q[\mathcal{A} \llbracket a \rrbracket / x]\} x=a \{Q|C_{ASSIGN\_V}\}$, $P \rightarrow Q[\mathcal{A} \llbracket a \rrbracket / x]$, and $T \ge C_{ASSIGN\_V}$, we have
$$\vdash \{ P \} x=a \{ Q | T \}$$

\textit{Case Array:}

$$VCG(\{P\} x[a_1]=a_2 \{Q|T\}) = \{P \rightarrow Q[\mathcal{A} \llbracket a_2 \rrbracket / x[\mathcal{A} \llbracket a_1 \rrbracket]] \} \cup \{T \ge C_{ASSIGN\_A}\}$$

By the \textit{array axiom} $\vdash \{Q[\mathcal{A} \llbracket a_2 \rrbracket / x[\mathcal{A} \llbracket a_1 \rrbracket]]\} x[a_1]=a_2 \{Q|C_{ASSIGN\_A}\}$.
By the weak rule, knowing $\{Q[\mathcal{A} \llbracket a_2 \rrbracket / x[\mathcal{A} \llbracket a_1 \rrbracket]]\} x[a_1]=a_2 \{Q|C_{ASSIGN\_A}\}$, $P \rightarrow Q[\mathcal{A} \llbracket a_2 \rrbracket / x[\mathcal{A} \llbracket a_1 \rrbracket]]$, and $T \ge C_{ASSIGN\_A}$, we have
$$\vdash \{ P \} x[a_1]=a_2 \{ Q | T \}$$

\textit{Case seq:}

Induction Hypothesis: 
$$ \models VCG(\{P\} S_1 \{R | T_1 \} \rightarrow \vdash  \{P\} S_1 \{R|T_1\})$$
$$ \models VCG(\{P\} S_2 \{R | T_2 \} \rightarrow \vdash \{P\} S_2 \{R|T_2\})$$

Let us consider:
\begin{itemize}
    \item $wp_2, t_2 = wpc(S_2,Q)$
    \item $wp_1, t_1 = wpc(S_1, wp_2)$
    \item $wpc(S1;S2 , Q) = (wp_1, t_1 + t_2)$
\end{itemize}

$$\models VCG(\{P\} S_1; S_2 \{R | T \} = \{P \rightarrow wp_1\} \cup \{T \ge t_1 + t_2\} \cup VC(S1;S2, R)$$
Where $VC(S_1;S_2, R) = VC(S_1, wp_2) \cup VC(S_2, R)$

Assuming $\models VCG(\{P\} S_1; S_2 \{R | T\}$

\begin{itemize}
    \item since we know $P \rightarrow wp_1$, $t_1 \ge t_1$, and $\models VC(S_1, wp_2)$, then $$\models VCG(\{P\} S_1 \{wp_2 | t_1 \}$$
    \item since we know $wp_2 \rightarrow wp_2$, $t_2 \ge t_2$, and $\models VC(S_2, R)$, then $$\models VCG(\{wp_2\} S_2 \{R | t_2 \}$$
\end{itemize}

By our Induction Hypothesis, we have $\vdash \{P\} S_1 \{wp_2|t_1\})$, and $\vdash \{wp_2\} S_2 \{R|t_2\})$.
By the seq rule
$$\vdash \{P\} S_1; S_2 ~\{R | t_1 + t_2\}$$
Since $T \ge t_1 + t_2$, by the weak rule
$$\vdash \{P\} S_1;S_2 \{R | T\}$$

\textit{Case if:}

Induction Hypothesis:
$$ \models VCG(\{P\} S_1 \{Q | t_1 \} \rightarrow \vdash \{P\} S_1 \{Q|t_1\})$$
$$ \models VCG(\{P\} S_2 \{Q | t_2 \} \rightarrow \vdash \{P\} S_2 \{Q|t_2\})$$

Let us consider
\begin{itemize}
    \item $wp_1, t_1 = wpc(S_1,Q)$
    \item $wp_2, t_2 = wpc(S_2,Q)$
    \item $wpc($if $b$ then $S_1$ else $S_2, Q) = (\mathcal{B} \llbracket b \rrbracket \rightarrow wp_1 \land \neg \mathcal{B} \llbracket b \rrbracket \rightarrow wp_2, max(t_1,t_2)+\mathcal{TB}\llbracket b \rrbracket)$
\end{itemize}

\begin{align*}
VCG(\{P\}\mathrm{if}\ b\ \mathrm{then}\ S_1\ \mathrm{else} S_2\{Q|T\}) = 
&\{P \rightarrow (b \rightarrow wp_1 \land \neg b \rightarrow wp_2)\}\ \cup \\
&\{T \ge max(t_1,t_2) + \mathcal{TB}\llbracket b \rrbracket\}\ \cup \\
&VC(\mathrm{if}\ b\ \mathrm{then}\ S_1\ \mathrm{else}\ S_2, Q)
\end{align*}

Where $VC(\mathrm{if}\ b\ \mathrm{then}\ S_1\ \mathrm{else} S_2, Q) = VC(S_1, Q) \cup VC(S_2,Q)$

Assuming $\models VCG(\{P\}\mathrm{if}\ b\ \mathrm{then}\ S_1\ \mathrm{else}\ S_2\{Q|T\})$

\begin{itemize}
    \item Since $P \land \mathcal{B} \llbracket b \rrbracket \rightarrow wp_1$, $t_1 \ge t_1$, and $VC(S_1, Q)$,  $$\models VCG(\{P \land \mathcal{B} \llbracket b \rrbracket\}S_1\{Q|t_1\})$$
    \item Since $P \land \neg \mathcal{B} \llbracket b \rrbracket \rightarrow wp_2$, $t_2 \ge t_2$, and $VC(S_2, Q)$,  $$\models VCG(\{P \land \neg \mathcal{B} \llbracket b \rrbracket\}S_2\{Q|t_2\})$$
\end{itemize}

From our Induction Hypothesis, we have $\vdash \{P \land \mathcal{B} \llbracket b \rrbracket\} S_1 \{Q|t_1\}$, and $\vdash \{P \land \neg \mathcal{B} \llbracket b \rrbracket\} S_2 \{Q|t_2\}$

By the if the rule, we get 
$$\{P\}\mathrm{if}\ b\ \mathrm{then}\ S_1\ \mathrm{else}\ S_2 \{Q|max(t_1,t_2) + \mathcal{TB} \llbracket b \rrbracket\}$$

Since $T \ge max(t_1,t_2) + \mathcal{TB} \llbracket b \rrbracket$, by the weak rule
$$\vdash \{P\}\mathrm{if}\ b\ \mathrm{then}\ S_1\ \mathrm{else}\ S_2\{Q|T\}$$

\textit{Case while:}

Induction Hypothesis:
$$\models VCG(\{P\}S\{Q|T\}) \rightarrow\ \vdash \{P\}S\{Q|T\}$$

Let us consider:
\begin{itemize}
    \item $wpc(\mathrm{while}\ b\ \mathrm{do}\ S, Q) = (I \land f \ge 0, \sum_{i=0}^{N-1}t(i) + (N+1) \times \mathcal{TB} \llbracket b \rrbracket)$
    \item $wp_S, t_S = wpc(S, I \land f > k)$
\end{itemize}
    \begin{align*}
        VC(\mathrm{while}\ b\ \mathrm{do}\ S, Q) = 
        &\{ (I \land \mathcal{B} \llbracket b \rrbracket \land f=k) \rightarrow wp_S \land t(k) \ge t_S \}\ \cup \\
        \{(I \land \neg b) \rightarrow Q\}\ \cup \\
        &\{(I \land B \llbracket b \rrbracket) \rightarrow f \le N\}\ \cup \\
        &VC(S, I \land f>k)
    \end{align*}

\begin{align*}
    VCG(\{P\} \mathrm{while}\ b\ \mathrm{do}\ S \{Q | T \}) = 
    &\{P \rightarrow (I \land f \ge 0)\}\ \cup\\
    &\{P \rightarrow T \ge \sum_{i=0}^{N-1} t(i) + (N + 1) \times \mathcal{TB} \llbracket b \rrbracket \}\ \cup\\
    &VC(\mathrm{while}\ b\ \mathrm{do}\ S, Q)
\end{align*}

Assuming $\models VCG(\{P\} \mathrm{while}\ b\ \mathrm{do}\ S \{Q | T \})$ 

Since $(I \land \mathcal{B} \llbracket b \rrbracket \land f=k) \rightarrow wp_S \land t(k) \ge t_S$, and $VC(S, I \land f>k)$ we have

$$ \models VCG(\{I \land \mathcal{B} \llbracket b \rrbracket \land f=k\}S\{I \land f > k | t(k)\})$$

By our Induction Hypothesis
$$\vdash \{I \land \mathcal{B} \llbracket b \rrbracket \land f=k\}S\{I \land f > k | t(k)\}$$

Since $I \land \mathcal{B} \llbracket b \rrbracket \rightarrow f \le N$, and $ \vdash \{I \land \mathcal{B} \llbracket b \rrbracket \land f=k\}S\{I \land f > k | t(k)\}$, by the \textit{while rule}

$$\vdash \{I \land f \ge 0\} \mathrm{while}\ b\ \mathrm{do}\ S \{I \land \mathcal{B} \llbracket b \rrbracket |\sum_{i=0}^{N-1}t(i) + (N+1) \times \mathcal{TB} \llbracket b \rrbracket\}$$

Since $P \rightarrow (I \land f \ge 0)$, $(I \land \neg \mathcal{B} \llbracket b \rrbracket) \rightarrow Q$, and $P \rightarrow T \ge \sum_{i=0}^{N-1}t(i) + (N+1) \times \mathcal{TB} \llbracket b \rrbracket$, by the \textit{weak rule}

$$\vdash \{P\}\mathrm{while}\ b\ \mathrm{do}\ S\{Q|T\}$$
\end{proof}

\subsection*{Example}\label{chap4:subsec:vcg-example}

Let us apply the VCG algorithm to the division example~\ref{fig:division}. We will use the same notation as in section~\ref{chap4:sec:axiomatic} and refer to our precondition as $P$, our program as $S$, our postcondition as $Q$ and our cost upperbound as $T$. The algorithm starts with a call to the $VCG$ function.

\begin{equation}
VCG(\{P\}S\{Q|T\}) = \{P \rightarrow wp)\ \cup \\
\{P \rightarrow T \ge t\}\ \cup \\
VC(S,Q) \nonumber
\end{equation}

Where
$$wp, t = wpc(S,Q) = (I \land x-r \ge 0, \sum_{i=0}^{x}10 + (x+1)\times 3)$$

Then the $VCG$ function calls the VC function for $S$. Since only \textit{while} loops generate extra VCs, we will omit other calls to the $VC$ function for simplicity.
\begin{align}
VC(S,Q) = &\{(I \land y \le r \land x-r=k) \rightarrow wp_S\}\ \cup \\
&\{(I \land \neg (y \le r)) \rightarrow Q\}\ \cup \\
&\{(I \land y \le r) \rightarrow x-r \le x\}\ \cup \\
&\{(I \land y \le r \land x-r=k) \rightarrow t(k) \ge t_S\}\ \cup \\
&VC(r=r-y;q=q+1,I \land x-r>k) \nonumber
\end{align}

Where 
\begin{align*}
wp_S, t_S &= wpc(r=r-y;q=q+1, I \land x-r>k) \\
&= (wp_1, t_1 + t_2) \\\\
wp_1, t_1 &= wpc(r=r-y, wp_2) \\
&= (wp_2[r-y/r], \mathcal{TA} \llbracket r-y \rrbracket + 1) \\\\
wp_2, t_2 &= wpc(q=q+1, I \land x-r>k) \\
&= ((I \land x-r>k)[q+1/q],\mathcal{TA} \llbracket q + 1 \rrbracket + 1)
\end{align*}

To prove our triple, we now simply need to prove all of the VCs generated by the algorithm (4.3 to 4.8), this can easily be done for all the conditions manually, or with the assistance of a theorem prover. 

As would be expected, proving a Hoare triple by applying the VCG algorithm is simpler and more mechanic than proving it directly by applying our rules and deriving the inference tree.

\section{Amortized Costs}\label{sec:amortized}
In chapter~\ref{sec:logic} we introduced our language, a cost-aware operational semantics, an axiomatic semantics to verify cost upper-bounds, and a VCG algorithm to apply our logic. In this chapter, we present a variation to our logic by considering amortized analysis to refine the estimation of costs for the \textit{while} loop. We start by giving some background on amortization before presenting our updated definitions and some examples. 

\section{Background}
\textit{Amortized analysis} is a method defined by Tarjan for analyzing the complexity of a sequence of operations~\cite{tarjan85}.
Instead of reasoning about the worst-case cost of individual operations, amortized analysis concerns the worst-case cost of a sequence of operations. The advantage of this method is that some operations can be more expensive than others. Thus, distributing the cost of expensive operations over the cheaper ones can produce better bounds than analyzing the worst case for individual operations.

Let $a_i$ and $t_i$ be, respectively, the amortized cost and the actual cost of each \textit{i} operation. In order to obtain an amortized analysis, it is necessary to define an \emph{amortized cost} such that
 $$ \sum_{i=1}^{n} a_i \geq \sum_{i=1}^{n} t_i $$
i.e.\@, for a sequence of $n$ operations, the  total amortized cost is an upper bound on the total actual cost. 

Thus for each intermediate step, the accumulated amortized cost is an upper bound on the accumulated actual cost. This allows the use of operations with an actual cost that exceeds their amortized cost. Conversely, other operations have a lower actual cost than their amortized cost. Expensive operations can only occur when the difference between the accumulated amortized cost and the accumulated actual cost is enough to cover the extra cost. 

There are three methods for amortized analysis: the \emph{aggregate method}, the \emph{accounting method}, and the \emph{potential method}. 
Let us analyze each of the methods with the same example.

\paragraph{Dynamic Array}
Let us consider a simple dynamic array algorithm. We will perform n insertions on the array. Every time the array is full, we create a new array with double the size, and all the elements must be copied to the new array. 
Consider a worst-case analysis of this algorithm for n insertions. When inserting element $i$, we might need to resize, so one insertion might lead to a resize, which would copy $i-1$ elements to the new array, plus the cost of inserting $i$, which means the worst-case scenario cost of one insertion is $i$. The cost of inserting $n$ elements would be $\sum_{i=1}^{n}i$, which is $\mathcal{O}(n^2).$

\paragraph{Aggregate Method}
The aggregate method considers the worst-case execution time $T(n)$ to run a sequence of $n$ operations. The amortized cost for each operation is $T(n)/n$. 
Applying the aggregate method to our dynamic array example gives us an amortized cost of $\frac{\mathcal{O}(n^2)}{n} = \mathcal{O}(n)$ per insertion. 

\paragraph{Accounting Method}
The Accounting Method, sometimes also called the taxation method, is a method where we tax cheaper operations, so we always have enough saved up to cover more expensive operations without ever going out of credit. 
To apply the accounting method to the dynamic array, we need to define what are the cheap operations we need to tax. Let us consider we have already inserted $m$ elements. Inserting an element in an array is $T(1)$ if the array is not full. If the array is full, we create a new array with size $2m$ and copy all $m$ elements to this new array. In this case, the cost of inserting an element is $T(m)$ to copy all the elements, plus $T(1)$ to insert the new element. If we consider the cost of inserting as $T(3)$, $1$ being the actual insertion cost we will never run out of credit. For every array state, all elements after position $m/2$, have never been copied, so they still have 2 extra credits. If we need to resize again, 1 of these credits will be used to copy the element, the remainder $m/2$ credits will be used to pay for copying elements before position $m/2$ that might have run out of credit. Since we always double the array size, this will always be enough. Therefore the amortized cost of insertion is $3$, i.e. $\mathcal{O}(1)$.

\paragraph{Potential method} 
The Potential Method considers a function $\Phi$, which maps a data structure's state $d_i$ to a real number. While with the accounting method, we would tax operations, with the potential method, credit is associated with the state of the data structure. Let $d_i$ represent the state of the data structure after $i$ operations, and $\Phi_i$ represent its potential. A valid potential function guarantees two properties: the initial potential is 0, $\Phi_0 = 0$; the potential always remains positive, $\forall i. \Phi_i \ge 0$.
The amortized cost of an operation $a_i$ is defined as its actual cost $t_i$, plus the change in potential between $d_{i-1}$ and $d_i$:
$$ a_i = t_i + \Phi_i - \Phi_{i-1} $$ 
From this, we get for \textit{j} operations:
\begin{equation}
\sum_{i=1}^{j} t_i = \sum_{i=1}^{j}(a_i + \Phi(d_{i-1}) - \Phi(d_i)) \\
\sum_{i=1}^{j} t_i = \sum_{i=1}^{j}a_i + \sum_{i=1}^{j}(\Phi(d_{i-1}) - \Phi(d_i))
 \end{equation}
The sequence of potential function values forms a telescoping series; thus, all terms except the initial and final values cancel. 
$$ \sum_{i=1}^{j}t_i = \sum_{i=1}^{j}a_i + \Phi(d_0) - \Phi(d_j)$$
Since $\Phi_0 = 0$ and  $\Phi_j \ge 0$  then 
$$\sum_{i=1}^{j} a_i \ge \sum_{i=1}^{j} t_i$$
Let us apply the potential method to the dynamic array problem. We start by defining our potential function $\phi$ as 2 times the number of elements after position $m/2$ for an array of size m, $\Phi_n = 2 \times (n - m/2) = 2n-m$. If the array is not full, the real cost ($t_i$) of inserting $i$ in the array is 1, and the change in potential $\phi_i - \phi_{i-1}$ is $2$, which means the amortized cost of this insertion is $3$.  If the array is full the real cost ($t_i$) is the cost of copying $m$ elements plus the cost of one insertion, $t_i = m + 1$. The change in potential is $2 - m$, then the amortized cost $a_i$ will be $3$, i.e. $\mathcal{O}(1)$. This method, like the accounting method, gives us a better amortized cost than the aggregated method.

The choice of method depends on how convenient each method is to the situation. The proof rules we will show in the next section use the potential method.
As we have seen in chapter~\ref{sec:related}, amortized analysis based on the potential method was already used in previous work, which derives upper bounds on the use of computational resources \cite{carbonneaux2015,hoffmann2006,simoes2012,vasconcelos2015}. Here we use it to prove tighter bounds when the composition of worst-case execution times is overly pessimistic.

\section{Proof Rules with Amortized Costs}\label{chap5:sec:axiomatic}
We now present our modified logic for amortized analysis. 
To this end, we modify the {\em while} rule (Figure~\ref{fig:hoare-amortized-while}) to allow deriving more precise bounds. Similarly to the \textit{while rule} presented in figure~\ref{fig:hoare}, we still get the variant $f$, the invariant $I$, and the maximum number of iteration $N$ from the oracle.
 But now, this new rule requires additional information from the oracle: an amortized cost for each iteration $a$ and a potential function $\phi$.
\begin{figure}[htbp]
  \centering
  \includegraphics[width=0.85\linewidth]{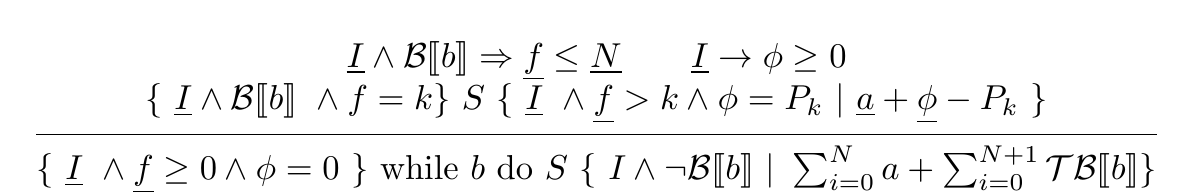}
  \caption{ Hoare rule for \textit{while} statement with amortized costs.}
  \label{fig:hoare-amortized-while}
\end{figure}
Regarding these new values, we add the following restrictions:
\begin{itemize}
    \item The potential function $\phi$ must be zero before we start the \textit{while}, $\phi = 0$
    \item The potential function $\phi$ must remain positive after every iteration of the \textit{while} $I \rightarrow \phi \ge 0$
    \item Knowing that $\langle S, \sigma \rangle \rightarrow^t \sigma_1$, we must have that $\forall k. \mathcal{A} \llbracket a + \phi - P_k \rrbracket \sigma \ge t$
\end{itemize}
Note that we have added a logic variable $P_k$. This variable is used in the time assertion $a + \phi - P_k$, and allows us to refer to the value of variable $\phi$ in two different states, before and after executing statement $S$.

\subsection*{Soundness}

Assume $I \land \mathcal{B} \llbracket b \rrbracket \Rightarrow f \le N$, $I \Rightarrow \phi \ge 0$, and 
$$\models \{\ I \land \mathcal{B} \llbracket b \rrbracket\ \land f=k\}\ S\ \{\ I\ \land \phi = P_{k} \land f > k\ |\ a + \phi - P_{k}\ \}$$

Considering a state $\sigma$ such that $\sigma \models I\ \land\ f \ge 0\ \land\ \phi = 0$ and $\langle \mathrm{while}\ b\ \mathrm{do}\ S, \sigma\rangle \rightarrow^{t_w} \sigma_1$.

Given our assumption, we know that $\sigma \models f \ge 0$ and that every time we enter the \textit{while} body, f increases. We also know that as long as $I \land B[b]$ are true, $f \le N$. Therefore  we know that the program will eventually stop and at most iterate $N+1$ times.

If $\sigma \models \neg \mathcal{B} \llbracket b \rrbracket$ then $\sigma_1 = \sigma$, therefore $\sigma_1 \models \neg \mathcal{B} \llbracket b \rrbracket \land I$. In this case $t_w = \mathcal{T}\mathcal{B}\llbracket b \rrbracket$. Since $N \ge 0$, then $\mathcal{A} \llbracket \sum_{i=0}^{N}a + \sum_{i=0}^{N+1} \mathcal{TB} \llbracket b \rrbracket \rrbracket \sigma \ge t_w$.

If $\sigma \models \mathcal{B} \llbracket b \rrbracket$, we have that $\sigma \models \mathcal{B} \llbracket b \rrbracket \land I$.
Considering a state $\sigma_2$ such that $\langle S, \sigma \rangle \rightarrow^{t_1} \sigma_2$ and $\langle \mathrm{while}\ b\ \mathrm{do}\ S, \sigma_2\rangle \rightarrow^{t_2} \sigma_1$. Applying our initial assumption we get $\sigma_2 \models I \land \phi = P_{k} \land f > k$.
Finally by applying the induction hypothesis we have that $\sigma_1 \models \neg \mathcal{B} \llbracket b \rrbracket \land I$.
\\

By our assumption we also know that $\forall k. \mathcal{A} \llbracket(a + \phi - P_{k})\rrbracket \sigma \ge t_1$. 

By induction, we know that 
$\langle \mathrm{while b do S} , \sigma_2 \rangle \rightarrow ^{t_2} \sigma_1$
where 
$$\mathcal{A} \llbracket (N-1) \times a + N \times \mathcal{T} \mathcal{B} \llbracket b \rrbracket \rrbracket \sigma_2 \ge t_2$$

The real cost of the \textit{while} is given by $t_1 + t_2 + \mathcal{T}\mathcal{B}\llbracket b \rrbracket$.

$$\mathcal{A} \llbracket(a + \phi - P_{k})\rrbracket \sigma \ge t_1$$

Since $\sigma \models \phi = 0 \rightarrow \mathcal{A} \llbracket \phi \rrbracket \sigma = 0$. And since $\sigma_2 \models I$ and $I \rightarrow \phi \ge 0$, then $\sigma \models \phi \ge 0$. $P_k = \mathcal{A} \llbracket \phi \rrbracket \sigma_2$, then $P_k \ge 0$. Then we have

$$\mathcal{A} \llbracket a \rrbracket \sigma \ge t_1$$

Since $a$, $N$, and $\mathcal{TB} \llbracket b \rrbracket$ are all constant, we get

\begin{equation}
\mathcal{A} \llbracket N \times a + N \times \mathcal{T} \mathcal{B} \llbracket b \rrbracket \rrbracket \sigma \ge t_1 + t_2 \\
\mathcal{A} \llbracket N \times a + (N+1) \times \mathcal{T} \mathcal{B} \llbracket b \rrbracket \rrbracket \sigma \ge t_1 + t_2 + \mathcal{T} \mathcal{B} \llbracket b \rrbracket
\end{equation}

Therefore the rule for \textit{while} is sound.

$$\models \{\ I\ \land f \ge 0\ \}\ \mathrm{while}\ b\ \mathrm{do}\ S\ \{\ I \land \neg \mathcal{B} \llbracket e \rrbracket\ |\ N \times a + (N+1) \times \mathcal{T}\mathcal{B}\llbracket b \rrbracket\}$$

\subsection*{Example: Binary Counter}

We now illustrate this new version of the Hoare Logic on a typical application of amortized analysis: a binary counter.
In the binary counter algorithm, we represent a binary number as an array of zeros and ones. We start with an array with every value at zero, and with each iteration, we increase the number by one until we reach the desired value. Our implementation can be seen in figure~\ref{fig:binary-counter}. 

\begin{figure}[htbp]
  \centering
  \includegraphics[width=0.8\linewidth]{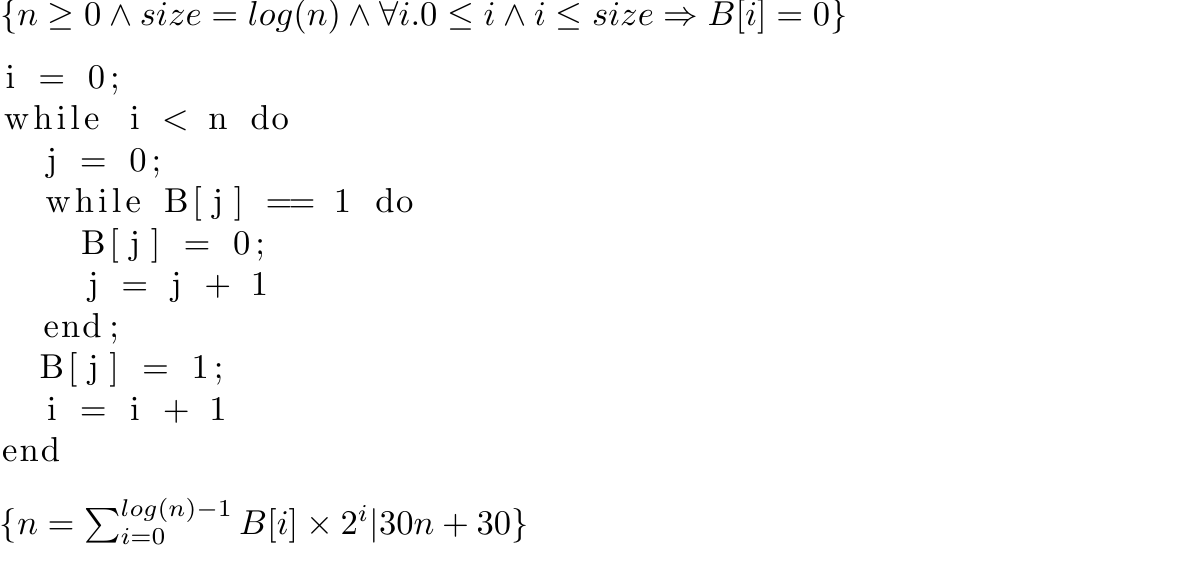}
  \caption{Binary Counter Implementation with Annotation.}
  \label{fig:binary-counter}
\end{figure}

To start our proof we will define the oracle information for each of the two \textit{while} loops. For simplicity, we refer to the external \textit{while} loop as \textit{while} 1, and the internal as \textit{while} 2.

The invariant of \textit{while} 1, $I_1 \equiv i = \sum_{k=0}^{size-1} B[k] \times 2^k \land 0 \le i \le n$, the variant $f_1 \equiv i $, the maximum number of iterations $N_1 \equiv n$, the amortized cos $a_1 \equiv 20$, and the potential function $ \phi_1 \equiv fun\ k \rightarrow \sum_{i=0}^{size} B[i]$. 
Here log(n) is the base two logarithm of n.

For the \textit{while} 2 we have $I_2 \equiv \forall k. 0 \le k < j \rightarrow B[k] = 0 \land 0 \le j \le size$, $f_2 \equiv j$, $N_2 \equiv size$, $a_2 \equiv 10$, and $ \phi_2 \equiv fun\ k \rightarrow 1$.

By the \textit{assign rule}, we have
$$\vdash \{(I_1 \land f_1 \ge k \land \phi_1 = P_k)[i+1/i]\} i = i + 1 \{ I_1 \land f_1 \ge k \land \phi_1 = P_k | 4 \}$$

By the \textit{assign rule}, we have
$$\vdash \{(I_1 \land f_1 \ge k \land \phi_1 = P_k)[i+1/i][1/B[j]]\} B[j] = 1 \{ (I_1 \land f_1 \ge k \land \phi_1 = P_k)[i+1/i] | 3 \}$$

Applying the \textit{seq rule}, we get
$$\vdash \{(I_1 \land f_1 \ge k \land \phi_1 = P_k)[i+1/i][1/B[j]]\} B[j] = 1; i = i + 1 \{ I_1 \land f_1 \ge k \land \phi_1 = P_k | 7 \}$$

By the \textit{assign rule}, we get
$$\vdash \{(I_2 \land f_2 \ge k \land \phi_2 = P_k)[j+1/j]\} j = j+1 \{ I_2 \land f_2 \ge k \land \phi_2 = P_k | 4 \}$$

By the \textit{assign rule}, we get
$$\vdash \{(I_2 \land f_2 \ge k \land \phi_2 = P_k)[j+1/j][0/B[j]]\} B[j]=0 \{ (I_2 \land f_2 \ge k \land \phi_2 = P_k)[j+1/j] | 3 \}$$

By the \textit{seq rule}
$$\vdash \{(I_2 \land f_2 \ge k \land \phi_2 = P_k)[j+1/j][0/B[j]]\} B[j]=0;j=j+1 \{ I_2 \land f_2 \ge k \land \phi_2 = P_k | 7 \}$$

Since we know $I_2 \land B[j]=1 \land f_2=k \rightarrow (I_2 \land f_2 \ge k \land \phi_2 = P_k)[j+1/j][0/B[j]] \land a_2 + \phi_2 - Pk \ge 7$, then by the \textit{weak rule}
$$\vdash \{I_2 \land B[j]=1 \land f_2=k\} B[j]=0;j=j+1\{I_2 \land f_2 > k \land \phi_2 = P_k | a_2 + \phi_2 - P_k\}$$

Then by the \textit{while rule}, we get
$$\vdash \{I_2 \land f_2 \ge 0 \} while_2 \{ I_2 \land \neg(B[j] = 1) | N_2 \times a_2 + (N_2 +1) \times \mathcal{TB} \llbracket B[j] = 1 \rrbracket \}$$

Since $I_2 \land \neg (B[j] = 1) \rightarrow (I_1 \land f_1 \ge k \land \phi_1 = P_k)[i+1/i][1/B[j]]$, then by the \textit{weak rule}
$$\vdash \{I_2 \land f_2 \ge 0 \} while_2 \{ (I_1 \land f_1 \ge k \land \phi_1 = P_k)[i+1/i][1/B[j]] | N_2 \times a_2 + (N_2 +1) \times \mathcal{TB} \llbracket B[j] = 1 \rrbracket \}$$

By the \textit{seq rule}
$$\vdash \{I_2 \land f_2 \ge 0 \} while_2;B[j]=1;i=i+1 \{ I_1 \land f_1 \ge k \land \phi_1 = P_k | N_2 \times a_2 + (N_2 +1) \times \mathcal{TB} \llbracket B[j] = 1 \rrbracket + 7 \}$$

By the\textit{ assign rule}, we have
$$\{(I_2 \land f_2 \ge 0)[0/j]\}j=0\{I_2 \land f_2 \ge 0 | 2\}$$

By the \textit{seq rule}
$$\vdash \{(I_2 \land f_2 \ge 0)[0/j] \} j=0;while_2;B[j]=1;i=i+1 \{ I_1 \land f_1 \ge k \land \phi_1 = P_k | N_2 \times a_2 + (N_2 +1) \times \mathcal{TB} \llbracket B[j] = 1 \rrbracket + 9 \}$$

Since $I_1 \land i < n \land f_1 = k \rightarrow (I_2 \land f_2 \ge 0)[0/j]$, and $a_1 + \phi_1 - P_k \ge N_2 \times a_2 + (N_2 +1) \times \mathcal{TB} \llbracket B[j] = 1 \rrbracket + 9$, then by the \textit{weak rule}
$$\vdash \{I_1 \land i < n \land f_1 = k \} j=0;while_2;B[j]=1;i=i+1 \{ I_1 \land f_1 \ge k \land \phi_1 = P_k | a_1+ \phi_1 - P_k \}$$

Since $I_1 \land i < n \rightarrow f_1 \le N_1$, and $I_1 \rightarrow \phi_1 \ge 0$, then we can apply the \textit{while rule} and we get
$$\vdash \{I_1 \land f_1 \ge 0 \land \phi_1 = 0 \} while_1 \{I_1 \land \neg(i<n) | N_1 \times a_1 + (N_1 + 1) \times 3\}$$

By the \textit{assign rule }
$$\vdash \{(I_1 \land f_1 \ge 0 \land \phi_1 = 0)[0/i] \} i=0 \{I_1 \land f_1 \ge 0 \land \phi_1 = 0 | 2\}$$

Applying the \textit{seq rule} gives us
$$\vdash \{(I_1 \land f_1 \ge 0 \land \phi_1 = 0)[0/i] \} i=0;while_1 \{I_1 \land \neg(i<n) | 2 + N_1 \times a_1 + (N_1 + 1) \times 3 \}$$

Since $n \ge 0 \land size=log(n) \rightarrow (I_1 \land f_1 \ge 0 \land \phi_1 = 0)[0/i]$, $I_1 \land \neg(i<n) \rightarrow n=\sum_{i=0}^{log(n) - 1} B[i] \times 2^i$, and $40n + 10n + 30 \ge 2 + N_1 \times a_1 + (N_1 + 1) \times 3$, by the \textit{weak rule}
$$\vdash \{n \ge 0 \land size=log(n)\} i=0; while_1 \{n=\sum_{i=0}^{log(n) - 1} B[i] \times 2^i | 40n + 10n + 30\}$$

\section{Verification Conditions Generation with Amortized Costs}
Considering the extensions to our logic, as presented and explained in the last section, we need to modify our VCG accordingly. The only changes we need to perform in wpc and VC are for the while cases. Both these changes are shown in Figure~\ref{fig:vcg-amortized}.

\begin{figure}[htbp]
  \centering
  \includegraphics[width=0.8\linewidth]{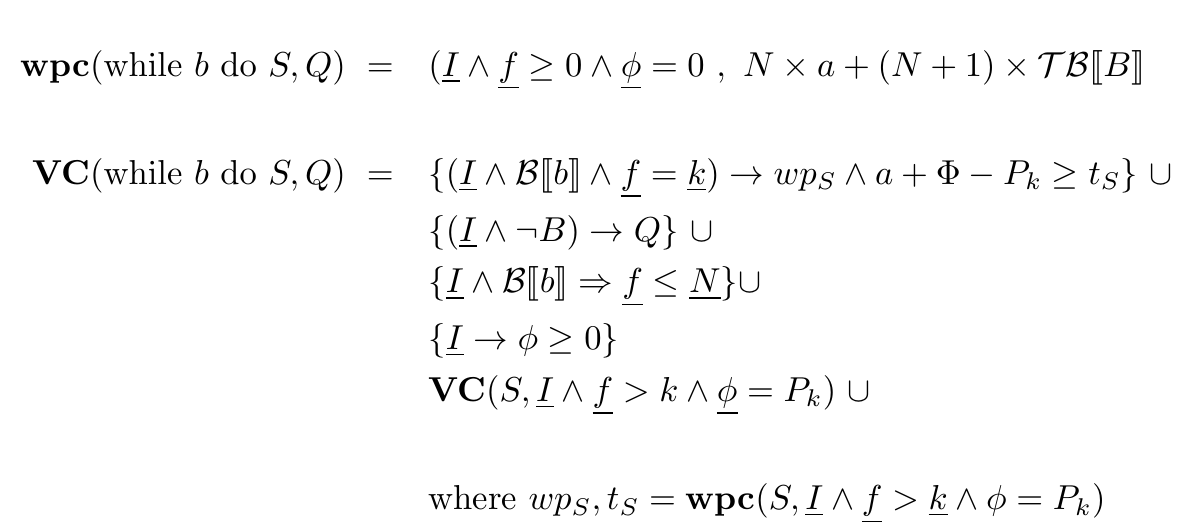}
  \caption{ VCG rules for \textit{while} statement with amortized costs.}
  \label{fig:vcg-amortized}
\end{figure}

The new wpc rule for \textit{while} returns a new term in the invariants conjunction that stipulates that $\phi = 0$, i.e., the potential function must be zero before the \textit{while} begins. The upper bound is given by the sum of the amortized cost for every iteration $\sum_{i = 0}^{N}a$, plus the sum of every $N+1$ evaluation of $b$, $\sum_{i=0}^{N+1}\mathcal{TB} \llbracket b \rrbracket$.

The VC rule for \textit{while} generates four VCs:
\begin{itemize}
    \item $\forall k.(I \land B \land f = k) \rightarrow wpc(S, I \land f > k \land \phi = P_k)$ guarantees both the invariant, the variant increase with every iteration, and defines a logic variable $P_k$ which allows us to refer to the state of variable $\phi$ in the postcondition. 
    \item $I \land \neg b \rightarrow Q$ which states that at the end of the \textit{while} (when $b$ is false), the invariant and $\neg b$ imply the postcondition $Q$ we wish to prove.
    \item $I \land \mathcal{B} \llbracket b \rrbracket \rightarrow f \le N$, which states that while the \textit{while} is still running (when $b$ is true), the variant is always less or equal to $N$.
    \item $I \rightarrow \phi \ge 0$, the potential function is always positive.
\end{itemize}
This rule also makes a recursive call to VC for the loop's body.

\subsection*{Soundness}

\begin{proof}

We start by defining our Induction Hypothesis.
$$\models VCG(\{P\}S\{Q|T\}) \rightarrow\ \vdash \{P\}S\{Q|T\}$$

We also calculate the result of wpc, VC and VCG for \textit{while}.

$$wpc(\mathrm{while}\ b\ \textit{do}\ S, Q) = (I \land f \ge 0 \land \Phi = 0, N \times a + (N+1) \times \mathcal{TB} \llbracket b \rrbracket)$$

$$wp_S, t_S = wpc(S,I \land f > k \land \phi = P_k)$$

\begin{align}
VC(\mathrm{while}\ b\ \textit{do}\ S, Q) = &\{ (I \land \mathcal{B} \llbracket b \rrbracket \land f = k) \rightarrow {wp}_S \} \cup \label{chap5:eq-1}\\
&\{I \land \mathcal{B} \llbracket b \rrbracket \land f = k \rightarrow a + \Phi - P_k \ge t_s \} \cup \label{chap5:eq-2}\\
&\{ I \land \neg \mathcal{B} \llbracket b \rrbracket \rightarrow Q \} \cup \label{chap5:eq-3}\\
&\{I \land \mathcal{B} \llbracket b \rrbracket \rightarrow f \le N \} \cup \label{chap5:eq-4}\\
&\{ I \rightarrow \Phi \ge 0 \} \cup \label{chap5:eq-5}\\
&VC(S,I \land f>k \land \Phi=P_k) \label{chap5:eq-6}
\end{align}

\begin{align}
VCG(\{P\}\mathrm{while}\ b\ \textit{do}\ S\{Q|T\}) = &\{P \rightarrow (I \land f \ge 0 \land \Phi=0)\} \cup \label{chap5:eq-7}\\
&\{P \rightarrow T \ge N \times a + (N+1) \times \mathcal{TB} \llbracket b \rrbracket \} \cup \label{chap5:eq-8} \\
&VC(\mathrm{while}\ b\ \textit{do}\ S, Q) \label{chap5:eq-9}
\end{align}

Assuming $\models VCG(\{P\}\mathrm{while}\ b\ \textit{do}\ S\{Q|T\})$.

Since we have \ref{chap5:eq-1}, \ref{chap5:eq-2}, and \ref{chap5:eq-6}, then
$$\models VCG(\{I \land \mathcal{B} \llbracket b \rrbracket \land f=k\} S \{I \land f>k \land \Phi = P_k | a + \Phi - P_k \})$$

By the induction hypothesis 
\begin{equation}
\vdash \{I \land \mathcal{B} \llbracket b \rrbracket \land f=k\} S \{I \land f>k \land \Phi = P_k | a + \Phi - P_k \}  \label{chap5:eq-10}
\end{equation}

Given \ref{chap5:eq-4}, \ref{chap5:eq-5}, and \ref{chap5:eq-10}, by the \textit{while rule}
$$\vdash \{I \land f \ge 0 \land \Phi = 0\} \mathrm{while}\ b\ \textit{do}\ S \{I \land \neg \mathcal{B} \llbracket b \rrbracket | N \times a + (N+1) \times \mathcal{TB} \llbracket b \rrbracket\}$$

Given \ref{chap5:eq-3}, \ref{chap5:eq-7}, and \ref{chap5:eq-8}, by the \textit{weak rule}
$$\vdash \{P\} \mathrm{while}\ b\ \textit{do}\ S \{Q | T\} $$

\end{proof}

\subsection*{Example: Binary Counter}
Let us apply the VCG algorithm to the binary counter example~\ref{fig:binary-counter}. We will use the same notation as in section~\ref{chap5:sec:axiomatic} and refer to our precondition as $P$, our program as $S$, our postcondition as $Q$ and our cost upperbound as $T$.

We start by calling the wpc function.

$$wpc(while_1,Q) = (I_1 \land I \ge 0 \land \phi_1 = 0, n \times a_1 + (n + 1) \times 3)$$

$$wpc(i=0,I_1 \land i \ge 0 \land \phi_1 = 0) = ((I_1 \land i \ge 0 \land \phi_1 = 0)[0/i], 2)$$

$$wp_S, t_S = wpc(S,Q) = ((I_1 \land i \ge 0 \land \phi_1 = 0)[0/i], 2 + n \times a_1 + (n + 1) \times 3)$$

Then we call the VC function for our program. Since this function only generates extra VCs for \textit{while} loops, we will omit other calls to the VC function for simplicity.

\begin{align}
VC(S,Q) &= VC(while_1, Q) = \nonumber \\
&\{ (I_1 \land i < n \land i = k) \rightarrow wp_1\}\ \cup \\
&\{ (I_1 \land i < n \land i = k) \rightarrow 30 + \phi_1 -P_k \ge t_1\}\ \cup \\
&\{ (I_1 \land \neg i<n) \rightarrow Q \}\ \cup \\
&\{ I_1 \land i<n \Rightarrow i \le n\} \cup\ \\
&\{  I_1 \rightarrow \phi_1 \ge 0 \}\ \cup \\
&\textbf{VC}(S_w, I_1 \land i > k \land \phi_1 = P_k) \nonumber
\end{align}

Where $S_w \equiv j=0;while_2;B[j]=1;i=i+1$, and
\begin{align*}
wp_1, t_1 &= \textbf{wpc}(S_w, I_1 \land i > k \land \phi_1 = P_k) \\
&= (I_2 \land j \ge 0 \land \phi_2 = 0)[0/j], size \times a_2 + (size + 1) \times \mathcal{TB} \llbracket B[j] = 1 \rrbracket
\end{align*}

\begin{align}
VC&(S_w,I_1 \land i > k \land \phi_1 = P_k) = \nonumber \\
VC&(while_2, (I_1 \land i > k \land \phi_1 = P_k)[i+1/i][1/B[j]]) = \nonumber \\
&\{ (I_2 \land B[j]=1 \land j = k) \rightarrow wp_{2} \}\ \cup \\
&\{ (I_2 \land B[j]=1 \land j = k) \rightarrow a_2 + \phi_2 - P_K \ge t_2 \}\ \cup \\
&\{ (I_2 \land \neg B[j]=1) \rightarrow (I_1 \land i > k \land \phi_1 = P_k)[i+1/i][1/B[j]] \}\ \cup \\
&\{ I_2 \land B[j]=1 \Rightarrow j \le size\} \cup\ \\
&\{  I_2 \rightarrow \phi_2 \ge 0 \}\
\end{align}

Where
\begin{align*}
wp_{2}, t_{2} &= \textbf{wpc}(B[j]=0;j=j+1, I_2 \land j > k \land \phi_2= P_k) \\
&= ((I_2 \land j \ge 0 \land \phi_2 = 0)[j+1/j][0/B[j]], size \times a_2 + (size+1) \times 3 + 9)
\end{align*}

Finally we call the VCG function

\begin{equation}
    VCG(\{P\}S\{Q|T\}) = \{ P \rightarrow wp_S \}\ \cup \\
    \{ P \rightarrow T \ge t_S \}\ \cup \\
     VC(S,Q) \nonumber
\end{equation}

To prove our triple, we now simply need to prove all of the VCs generated by the algorithm (5.11 to 5.21), this can easily be done for all the conditions manually, or with the assistance of a theorem prover. 

As would be expected, proving a Hoare triple by applying the VCG algorithm is simpler and more mechanic than proving it directly by applying our rules and deriving the inference tree.

\section{Exact Costs}\label{sec:exact}
The final part of our work consists of extending our language and logic so that we can verify the exact cost of a program instead of an upper bound like in the previous chapters. This approach is advantageous in scenarios where the approximation of execution time is not enough to guarantee safety, as it happens for critical systems and real-time programming. 
Moreover, proving that the exact execution time of a program is an expression that does not depend on confidential data provides a direct way to prove the absence of timing leakage, which is relevant in cryptographic implementations.
We must restrict the programming language to guarantee the ability to prove exact costs. Thus, in this third scenario, programs have bound recursion (\textit{for} loops), and both branches of conditional statements must have identical costs.

\section{Operational Semantics for Exact Costs}\label{chap6:sec:operational}

Since we need to be able to calculate exact costs, we need to make some extensions to our language so that the execution time is fully deterministic. The first one is a restriction to our conditional statement, $\mathrm{if}\ b\ \mathrm{then}\ S_1\ \mathrm{else}\ S_2$. It is still possible to have conditional statements in this scenario. However, we need to guarantee that both branches of the \textit{if} will take the same time to run.

In our original language, we used \textit{while} loops, but since we can not predict accurately the exact amount of times a \textit{while} is going to run, we can not have this statement in this version. We will then replace our \textit{while} loop with a \textit{for} loop, which executes a deterministic number of times, solving our problem.

Let us now present our updated syntax rules for statements, which remain the same for every statement except the \textit{for} loop.

\begin{equation}
S ::= skip\ 
 |\ x = a\  
 |\ x [ a_1 ] = a_2\ 
 |\ \textbf{if}\ b\ \textbf{then}\ S_1\ \textbf{else}\ S_2\ \textbf{done}\ 
 |\ \textbf{for}\ i=a\ \textbf{to}\ b\ \textbf{do}\ S\ \textbf{done}\ 
 |\ S_1 ; S_2
\end{equation}

The semantic rules for \textit{for} loop are presented in figure~\ref{fig:for-semantic}.
We have one rule and one axiom for \textit{for} loop, $[for^{true}]$ and $[for^{false}]$.

If $\mathcal{B} \llbracket b \rrbracket \sigma$ is false, we apply the axiom $[for^{false}]$, that says we will remain in the same state $\sigma$ and the cost is simply the cost of the evaluation of $a < b$, $\mathcal{T}\mathcal{B} \llbracket a < b \rrbracket$. 

If $\mathcal{B} \llbracket b \rrbracket \sigma$  is true, we apply rule $[for^{true}]$, which means we will, assign $a$ to variable $i$, which will lead to a state $\sigma'''$, we will then execute the loop body, $S$, once from state $\sigma'''$ and this will lead to a state $\sigma''$. Finally, we execute the \textit{for} loop again, but this time starting at $a + 1$ and from state $\sigma''$. The cost of the \textit{for} loop, in this case, is the cost of evaluating $a < b$, plus the cost of executing the \textit{assign}, plus the cost of executing the body, plus the cost of executing the \textit{for} loop from state $\sigma''$.

\begin{figure}[htbp]
  \centering
  \includegraphics[width=\linewidth]{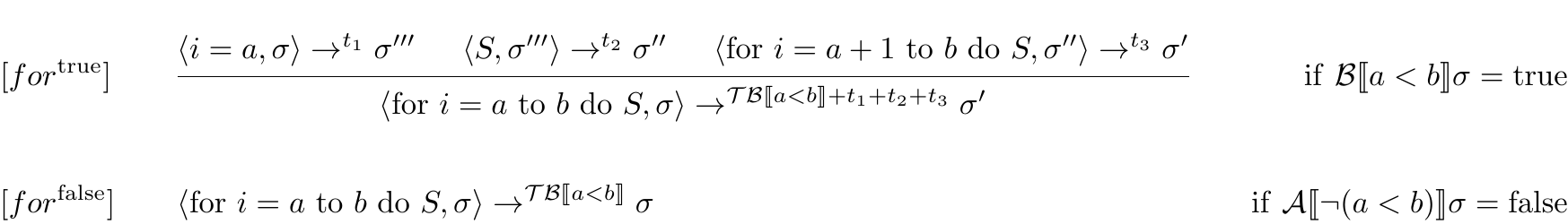}
  \caption{Operational semantic of \textit{for} loop.}
  \label{fig:for-semantic}
\end{figure}

The semantic rules for $skip$, $assign$, $array$, $seq$ and $if$ remain the same as presented in section~\ref{chap4:sec:formal-operational}.

\section{Proof Rules for Exact Costs}\label{chap6:sec:axiomatic}

In this new logic we have that $\models \{ P \} S \{ Q | t \}$
if and only if, for all state $\sigma$ such that $\sigma \models P$ and
$\langle S, \sigma \rangle \rightarrow^{t} \sigma'$ we have that $\sigma' \models Q$ and $\mathcal{A}\llbracket t \rrbracket\sigma = t'.$ Notice that this is fairly similar to what was presented in section~\ref{chap4:sec:axiomatic} but now, for a Hoare triple to be valid, the value passed in the cost section needs to represent the exact cost of the program ($\mathcal{A}\llbracket t \rrbracket\sigma = t'$).

The new axiomatic rules for the {\em for} loop are defined in figure~\ref{fig:hoare-for}.

\begin{figure}[htbp]
  \centering
  \includegraphics[width=\linewidth]{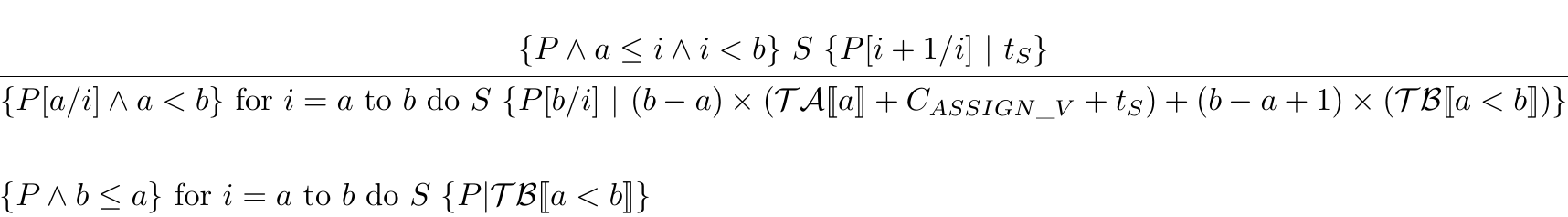}
  \caption{Hoare rule for \textit{for}-loop statement.}
  \label{fig:hoare-for}
\end{figure}

We have one rule and one axiom for $for$. The axiom says that if we are in a state that validates $P$ and $b \le a$, then after executing the \textit{for}, we will be in a state that validates P, and this execution will have an exact cost of $\mathcal{TB} \llbracket a < b \rrbracket$. 
The $for$ rule says that if we start in a state that validates $P$ and $a < b$ then after executing the \textit{for} loop we will be in a state that validates $P[b/i]$ and this will have a cost of $ (b - a + 1) \times \mathcal{TB} \llbracket a < b \rrbracket + (b-a) \times (t_S + \mathcal{TA}\llbracket a \rrbracket + C_{ASSIGN\_V)}$. To prove this, we must first prove that if we execute $S$ from a state that validates $P$ and $a \le i < b$, then after executing $S$, we will be in a state that validates $P[b/i]$ and this will cost $t_S$ to execute.

We also modify the rule for conditional statements by imposing that both branches must execute with the same exact cost. The rule for \textit{if} is then redefined as shown in Figure~\ref{fig:hoare-if}. Note that balancing \textit{if} branches with, e.g., dummy instructions is a common technique used in cryptography to eliminate execution time dependencies from branch conditions that may be related to secret data. The new rule is shown in Figure~\ref{fig:hoare-if}.

\begin{figure}[htbp]
  \centering
  \includegraphics[width=0.8\linewidth]{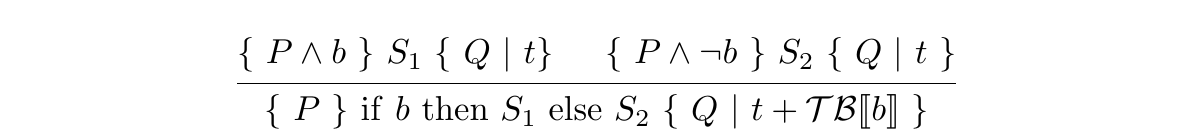}
  \caption{Hoare rule for \textit{if} statement where both branches take exactly the same time to execute.}
  \label{fig:hoare-if}
\end{figure}

The rule for $if$ says that, if we start at a state validating $P$ then after executing $if$ we will arrive to a state that validates Q. This execution will take exactly $t + \mathcal{TB} \llbracket b \rrbracket$. To prove this, we need to guarantee
\begin{itemize}
    \item Executing $S_1$ from a state validating $P \land b$ generates a state that validates $Q$ and $S_1$ takes $t$ to execute.
    \item Executing $S_2$ from a state validating $P \land \neg b$ generates a state that validates $Q$ and $S_2$ takes $t$ to execute.
\end{itemize}

\subsection*{Soundness}

We need to ensure that our Hoare logic is sound with respect to our operational semantic. For this version, our soundness theorem will be slightly different than the one we have previously presented. 

\begin{theorem}[Soundness]
We have that $\models \{ P \} S \{ Q | t \}$ if and only if, forall state $\sigma$ such that $\sigma \models P$ and $\langle S, \sigma \rangle \rightarrow^{t'} \sigma'$, we have $\sigma' \models Q$ and $\mathcal{A}\llbracket t \rrbracket\sigma = t'.$
\end{theorem}

Even though our theorem changed, the proof for \textit{skip}, \textit{assign}, \textit{array} and \textit{seq} will look exactly the same since the upper bound calculated by our previous logic was already identical to the real cost of execution. Therefore we will only show the proof for \textit{if} and \textit{for}. 

\begin{proof}
Case if:
Assume $\models \{ P \land b \} S_1 \{ Q | t \}$ and $\models \{P \land \neg b\} S_2 \{ Q | t\}$.
Suppose $\sigma \models P$. 

If $\sigma \models b$, then $\sigma \models P \land b$ so, assuming $\langle S_1, \sigma \rangle \rightarrow^{t_1} \sigma_1$, we have that $\sigma_1 \models Q$

If $\sigma \models \neg b$, then $\sigma \models P \land \neg b$ e so, assuming $\langle S_2, \sigma \rangle \rightarrow^{t_2} \sigma_2$, we have that $\sigma_2 \models Q$.

Given our assumptions we know that $t_1 = t_2 = t$. We then have that the exact cost for the \textit{if} statement is $t + \mathcal{T}\mathcal{B}\llbracket b \rrbracket$. The rule for \textit{if} is sound.

$$\models \{\ P\ \}\ \mathrm{if}\ e\ \mathrm{then}\ S_1\ \mathrm{else}\ S_2\ \{\ t_1 + \mathcal{T} \mathcal{B} \llbracket b \rrbracket\ |\ Q\ \}$$

Case for:
Assume $\models \{\ P \land (a \le i) \land i < b\ \}\ S\ \{\ P[i+1/i]\ | t_S\}$ \\
Suppose $\sigma \models P[a/i]$ and $\langle \mathrm{for}\ i=a\ \mathrm{to}\ b\ \mathrm{do}\ S, \sigma\rangle \rightarrow^{t} \sigma_1$.

If $\sigma \models \neg (a \le b)$ then $\sigma = \sigma_1$. From our lemma~\ref{lemma:subst} we get that $\sigma \models P$. In this case, the execution time is $\mathcal{T}\mathcal{B} \llbracket a \le b \rrbracket$. When $a > b$, $t = \mathcal{T}\mathcal{B} \llbracket a < b \rrbracket$, then the axiom is sound.

If $\sigma \models a< b$ then $\sigma \models P[a/i] \land a < b$.
Let us consider a state $\sigma_2$ such that  $\langle i=a, \sigma \rangle \rightarrow^{t_1} \sigma_2$. By the \textit{assign rule}, we know $\sigma_2 \models P$ and $t_1 = \mathcal{TA} \llbracket a \rrbracket + C_{ASSIGN\_V}$. We also have that $\sigma_2 \models a \le i < b$, therefore if we execute $S$ from state $\sigma_2$, we will get a state $\sigma_3$, $\langle S, \sigma_2 \rangle \rightarrow^{t_2} \sigma_3$, such that $\sigma_3 \models P[i+1/i]$, and $t_S = t_2$. By our induction hypothesis: $\langle \mathrm{for}\ i=a+1\ \mathrm{to}\ b\ \mathrm{do}\ S, \sigma_3 \rangle \rightarrow^{t_3} \sigma_1$, where $\sigma_1 \models P[b/i]$, and $t_3 = (b - a) \times \mathcal{TB} \llbracket a < b \rrbracket + (b-a-1) \times (t_S + \mathcal{TA}\llbracket a \rrbracket + C_{ASSIGN\_V)}$.
The real cost of executing statement $for$ from state $\sigma$ is $t = t_1 + t_2 + t_3 + \mathcal{TB} \llbracket a < b \rrbracket$. Knowing $t_1 = \mathcal{TA} \llbracket a \rrbracket + C_{ASSIGN\_V}$, $t_2 = t_S$, and $t_3 = (b - a) \times \mathcal{TB} \llbracket a < b \rrbracket + (b-a-1) \times (t_S + \mathcal{TA}\llbracket a \rrbracket + C_{ASSIGN\_V)}$, gives us
$$t = \mathcal{TA} \llbracket a \rrbracket + C_{ASSIGN\_V} + t_S + (b - a) \times \mathcal{TB} \llbracket a < b \rrbracket + (b-a-1) \times (t_S + \mathcal{TA}\llbracket a \rrbracket + C_{ASSIGN\_V)} + \mathcal{TB} \llbracket a<b \rrbracket$$
$$t = (b - a + 1) \times \mathcal{TB} \llbracket a < b \rrbracket + (b-a) \times (t_S + \mathcal{TA}\llbracket a \rrbracket + C_{ASSIGN\_V)} + \mathcal{TB} \llbracket a<b \rrbracket$$

Therefore the rule for \textit{for} is sound.
\end{proof}

\subsection*{Example: Range Filter}

To illustrate our logic, we will apply our rules to the range filter algorithm, as presented in figure~\ref{fig:range-filter}.
This algorithm consists of a simple filter where, given an array ($a$) and a range [l..u], we use an auxiliary array ($b$) to filter the elements in $a$ that are within the range.

\begin{figure}[htbp]
  \centering
  \includegraphics[width=0.8\linewidth]{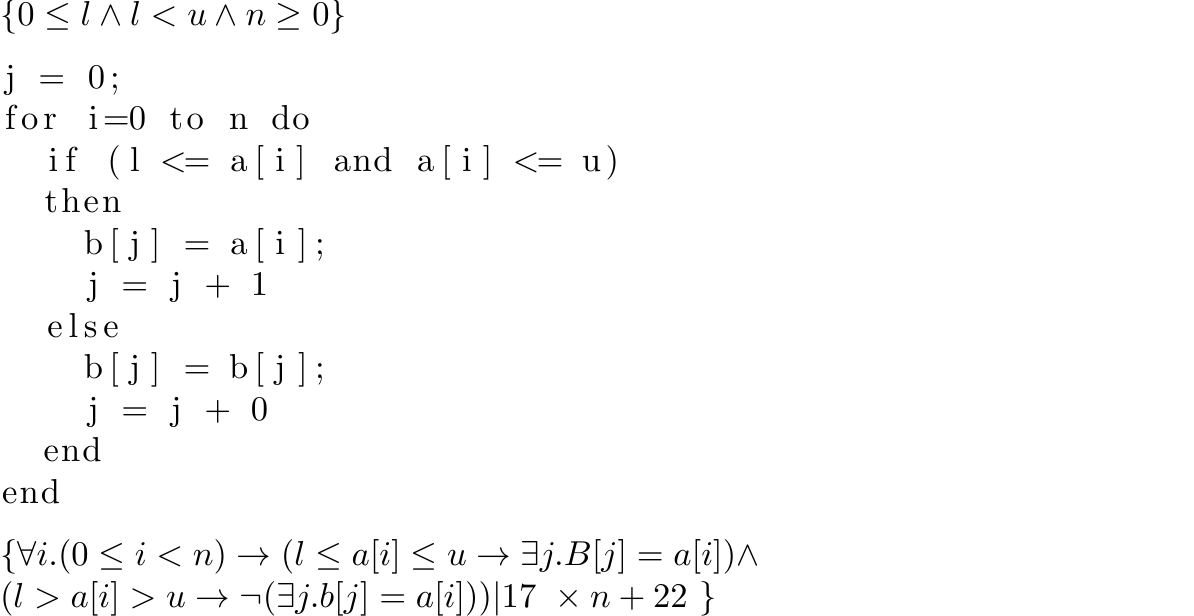}
  \caption{Array filtering algorithm with annotations for exact cost.}
  \label{fig:range-filter}
\end{figure}

We provide the invariant $I \equiv \forall k. (0 \le k \land k<i) \rightarrow (l \le a[k] \land a[k] \le u \rightarrow \exists k. b[j] = a[i] ) \land (l > a[k] \land a[k] > u \rightarrow \neg (\exists j. b[k] = a[i]))$ in order to prove correctness.

Let us refer to the \textit{if} statement as $S_{if}$ and to the \textit{for} statement as $S_{for}$.

By the \textit{assign rule} we have
$$\{I[i+1/i][j+1/j]\}j=j+1\{I[i+1/i]|\mathcal{TA} \llbracket j+1 \rrbracket + 1\}$$

By the \textit{assign rule} we also have
$$\{I[i+1/i][j+1/j][a[i]/b[j]]\}b[j]=a[i]\{I[i+1/i][j+1/j]|\mathcal{TA} \llbracket i \rrbracket + \mathcal{TA} \llbracket j \rrbracket + 1\}$$

By the \textit{seq rule} we get
$$\{I[i+1/i][j+1/j][a[i]/b[j]]\}b[j]=a[i]; j=j+1\{I[i+1/i]|4 + 3 + 1\}$$

By the \textit{assign rule} we have
$$\{I[i+1/i][j+0/j]\}j=j+0\{I[i+1/i]|\mathcal{TA} \llbracket j+0 \rrbracket + 1\}$$

By the \textit{assign rule} we also have
$$\{I[i+1/i][j+0/j][b[j]/b[j]]\}b[j]=b[j]\{I[i+1/i][j+0/j]|\mathcal{TA} \llbracket j \rrbracket + \mathcal{TA} \llbracket j \rrbracket + 1\}$$

By the \textit{seq rule} we get
$$\{I[i+1/i][j+0/j][b[j]/b[j]]\}b[j]=b[j]; j=j+0\{I[i+1/i]|4 + 3 + 1\}$$

Since $I[0/i] \land l \le a[i] \le u \rightarrow I[i+1/i][j+1/j][a[i]/b[j]]$, $I[0/i] \land  \neg (l \le a[i] \le u) \rightarrow I[i+1/i][j+0/j][b[j]/b[j]]$, then by the \textit{if rule}

$$\{I[0/i]\} S_{if} \{ I[i+1/i] | 8 + \mathcal{TB} \llbracket l \le a[i] \land a[i] \le u \rrbracket \}$$

By the \textit{for rule} we get
$$\{I[0/i]\} S_{for} \{I[b/i] | n \times 17 + (n+1) \times \mathcal{TB} \llbracket i < n \rrbracket\}$$

By the \textit{assign rule}
$$\{I[0/i][0/j]\}j = 0\{\ I[b/i] | 2\}$$

By the \textit{seq rule}
$$\{I[0/i][0/j]\} S \{ I[i+1/i] | 8 + n \times 9 + (n+1) \times \mathcal{TB} \llbracket i < n \rrbracket + 2\}$$

Since $P \rightarrow I[0/i][0/j]$, $I[i+1/i] \rightarrow Q$, and $T \ge n \times 9 + (n+1) \times \mathcal{TB} \llbracket i < n \rrbracket + 2$, then by the \textit{weak rule}

$$\vdash \{P\} S \{Q | T\}$$

\section{Verification Conditions Generation for Exact Costs}\label{chap6:sec:vcg}

Given the extensions to our logic, we must rewrite our VCG accordingly.

In figure~\ref{fig:wp-exact} we present the new wpc algorithm for $if$ and $for$. The rules for $skip$, $assign$, $array$, and $seq$ remain the same since they already output exact costs.
\begin{figure}[hbp]
  \centering
  \includegraphics[width=\linewidth]{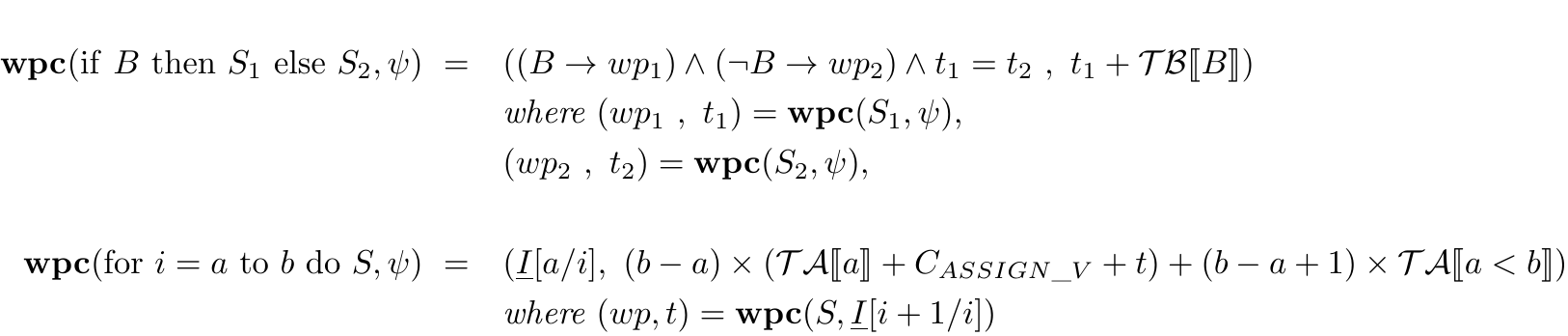}
  \caption{ Weakest Precondition Algorithm for Exact Costs.}
  \label{fig:wp-exact}
\end{figure}

In the $wpc$ result for $if$, we add a precondition restriction that says $t_1 = t_2$. In the cost expression, instead of computing the max between $t_1$ and $t_2$, we can simply define the cost as $t_1 + \mathcal{TB} \llbracket b \rrbracket$.

The wp of a \textit{for} loop is the invariant when $i=a$. The cost of executing a \textit{for} loop is $b-a$ times the cost of the loop body $t$, plus $b-a+1$ times the cost of evaluating $a < b$.

In figure~\ref{fig:vc-exact} we show the VC rules for $if$ and $for$. The rule for $if$ remains exactly the same as in the original version (\ref{fig:vc}). 

The rule for $for$ derives three VCs:
\begin{itemize}
    \item $I[b/i] \rightarrow Q$, meaning that when $i$ reaches value $b$ the loop breaks and the postcondition $Q$ is met. 
    \item $I \land a \le i < b \rightarrow wp(S, I)$, which guarantees the invariant is preserved and that before executing the \textit{for} loop body, $i$ must be a value between $a$ and $b$.
    \item $I \land \neg (a < b) \rightarrow Q$, which states that if $a < b$ is not met, then the loop will not execute and the postcondition $Q$ is true.
\end{itemize}
The rule also has a recursive call to VC applied to the body statement $S$.

\begin{figure}[htbp]
  \centering
  \includegraphics[width=0.8\linewidth]{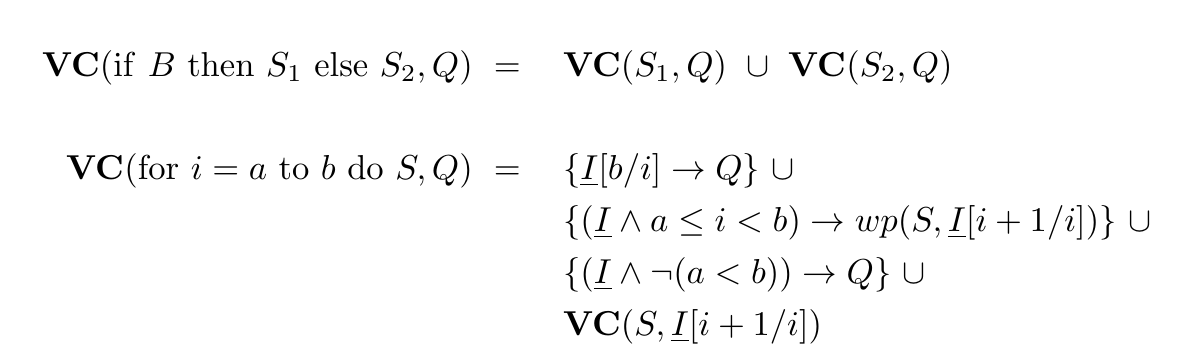}
  \caption{ VC Function for Exact Costs.}
  \label{fig:vc-exact}
\end{figure}

\subsection*{Soundness} 

We need to prove theorem~\ref{theor:soundness-vcg}, which states that the VCG algorithm is sound if the VC generated implies the Hoare triple we wish to prove, 
$$\models VCG(\{P\}Q\{R\})\ \Rightarrow\ \vdash \{P\}Q\{R\}$$

\begin{proof}
We prove $\Rightarrow$ by induction on the structure of Q and $\Leftarrow$ by induction in the derivation of $\vdash \{P\}Q\{R\}$.

\textit{Case if:}

Induction Hypothesis:
$$ \models VCG(\{P\} S_1 \{Q | t_1 \} \rightarrow \vdash \{P\} S_1 \{Q|t_1\})$$
$$ \models VCG(\{P\} S_2 \{Q | t_2 \} \rightarrow \vdash \{P\} S_2 \{Q|t_2\})$$

Let us consider
\begin{itemize}
    \item $wp_1, t_1 = wpc(S_1,Q)$
    \item $wp_2, t_2 = wpc(S_2,Q)$
    \item $wpc(\mathrm{if}\ b\ \mathrm{then}\ S_1\ \mathrm{else}\ S_2, Q) = (\mathcal{B} \llbracket b \rrbracket \rightarrow wp_1 \land \neg \mathcal{B} \llbracket b \rrbracket \rightarrow wp_2 \land t_1=t_2, t_1+\mathcal{TB}\llbracket b \rrbracket)$
\end{itemize}

\begin{align*}
VCG(\{P\}\mathrm{if}\ b\ \mathrm{then}\ S_1\ \mathrm{else} S_2\{Q|T\}) = &\{P \rightarrow (b \rightarrow wp_1 \land \neg b \rightarrow wp_2 \land t_1=t_2)\} \cup \\
&\{T = t_1 + \mathcal{TB}\llbracket b \rrbracket\} \cup \\
&VC(\mathrm{if}\ b\ \mathrm{then}\ S_1\ \mathrm{else} S_2, Q)
\end{align*}

Where $VC(\mathrm{if}\ b\ \mathrm{then}\ S_1\ \mathrm{else} S_2, Q) = VC(S_1, Q) \cup VC(S_2,Q)$.

Assuming $\models VCG(\{P\}\mathrm{if}\ b\ \mathrm{then}\ S_1\ \mathrm{else}\ S_2\{Q|T\})$.

\begin{itemize}
    \item Since $P \land \mathcal{B} \llbracket b \rrbracket \rightarrow wp_1$, $t_1 = t_1$, and $VC(S_1, Q)$,  $$\models VCG(\{P\}S_1\{Q|t_1\})$$
    \item Since $P \land \neg \mathcal{B} \llbracket b \rrbracket \rightarrow wp_2$, $t_2 = t_2$, and $VC(S_2, Q)$,  $$\models VCG(\{P\}S_2\{Q|t_2\})$$
\end{itemize}

From our Induction Hypothesis, we have $\vdash \{P \land \mathcal{B} \llbracket b \rrbracket\} S_1 \{Q|t_1\}$, and $\vdash \{P \land \neg \mathcal{B} \llbracket b \rrbracket\} S_2 \{Q|t_2\}$.

By the \textit{if rule}, we get 
$$\{P\}\mathrm{if}\ b\ \mathrm{then}\ S_1\ \mathrm{else}\ S_2 \{Q|t_1 + \mathcal{TB} \llbracket b \rrbracket\}$$

Since $T = t_1 + \mathcal{TB} \llbracket b \rrbracket$, by the \textit{weak rule}
$$\vdash \{P\}\mathrm{if}\ b\ \mathrm{then}\ S_1\ \mathrm{else}\ S_2\{Q|T\}$$

\textit{Case for:}

Induction Hypothesis:
$$ \models VCG(\{P\} S \{Q | T \} \rightarrow\ \vdash \{P\} S \{Q|T\})$$

Let us consider
$$wpc(\mathrm{for}\ i = a\ \mathrm{to}\ b\ \mathrm{do}\ S, Q) = (I[a/i], (b-a) \times (\mathcal{TA} \llbracket a \rrbracket + C_{ASSIGN\_V} + t)+ (b-a+1) \times \mathcal{TB}\llbracket a \le b \rrbracket)$$

where $wp_S, t_S = wpc(S, I[i+1/i])$

\begin{align}
VC(\mathrm{for}\ i = a\ \mathrm{to}\ b\ \mathrm{do}\ S,Q) = 
&\{I[b/i] \rightarrow Q \}\ \cup \label{chap6:eq-1}\\
&\{(I \land \neg(a < b)) \rightarrow Q\}\ \cup \label{chap6:eq-2}\\
&\{(I \land a \le i < b) \rightarrow wp\}\ \cup \label{chap6:eq-3}\\
&VC(S,I[i+1/i]) \label{chap6:eq-4}
\end{align}

\begin{align}
VCG(\{P\}\mathrm{for}\ i = a\ &\mathrm{to}\ b\ \mathrm{do}\ S\{Q|T\}) = \{P \rightarrow I[a/i]\}\ \cup \label{chap6:eq-5}\\
&\{T = (b-a) \times t+ (b-a+1) \times \mathcal{TB}\llbracket a \le b \rrbracket\}\ \cup \label{chap6:eq-6}\\
&VC(\mathrm{for}\ i = a\ \mathrm{to}\ b\ \mathrm{do}\ S, Q) \label{chap6:eq-7}
\end{align}

Assuming $\models VCG(\{P\}\mathrm{for}\ i = a\ \mathrm{to}\ b\ \mathrm{do}\ S\{Q|T\})$.

Given \ref{chap6:eq-3}, , $t_S = t_S$, and VC(S,I[i+1/i]) then $\models VCG[\{I\land a \le i < b\}S\{I[i+1/i]|t\}]$

From our Induction Hypothesis
$$\vdash \{I\land a \le i < b\}S\{I[i+1/i]|t\}$$

By the \textit{for} rule
$$\vdash \{I[a/i]\}\mathrm{for}\ i = a\ \mathrm{to}\ b\ \mathrm{do}\ S\{I[b/i] | (b-a) \times (\mathcal{TA} \llbracket a \rrbracket + C_{ASSIGN\_V} + t)+ (b-a+1) \times \mathcal{TB}\llbracket a \le b \rrbracket\}$$

Given \ref{chap6:eq-5}, \ref{chap6:eq-1} and \ref{chap6:eq-6} we get

$$\vdash \{P\}\mathrm{for}\ i = a\ \mathrm{to}\ b\ \mathrm{do}\ S\{Q|T\}$$
\end{proof}

\subsection*{Example: Range Filter}

We now apply the VCG algorithm to the range filter example~\ref{fig:range-filter}. 
We will use the same notation as in section~\ref{chap6:sec:axiomatic} and refer to our precondition as $P$, our program as $S$, our postcondition as $Q$, and our tight cost as $T$. We will also refer to the \textit{for} loop body statement as $S_{if}$. The algorithm starts with a call to the VCG function.

\begin{align}
VCG(\{P\}S\{Q|T\}) = &\{P \rightarrow wp)\ \cup \label{chap6:eq-8}\\
&\{T = t\}\ \cup \label{chap6:eq-9}\\
&VC(S,Q) \nonumber
\end{align}

Where
$$wp, t = wpc(S,Q) = (I[0/i][0/j] , n \times t_{if} + (n+1) \times \mathcal{TB}\llbracket 0 < n \rrbracket + 2)$$
$$wp_{if}, t_{if} = wpc(S_{if}, I[i+1/i])$$

Then the VCG function calls the VC function for $S$. Since only for loops generate extra VCs, we will omit other calls to the VC function for simplicity. 
\begin{align}
VC(S,Q) = 
&\{(I \land 0 \le i < n) \rightarrow wp_{if}\}\ \cup \label{chap6:eq-10}\\
&\{(I[n+1/i]) \rightarrow Q\}\ \cup \label{chap6:eq-11}\\
&\{(I[0/i] \land n < 0) \rightarrow Q\}\ \cup \label{chap6:eq-12}\\
&VC(S_f, I[i+1/i]) \nonumber
\end{align}

To prove our triple, we now simply need to prove all of the VCs generated by the algorithm (\ref{chap6:eq-8} to \ref{chap6:eq-12}), this can easily be done for all the conditions manually, or with the assistance of a theorem prover. 

As would be expected, proving a Hoare triple by applying the VCG algorithm is simpler and more mechanic than proving it directly by applying our rules and deriving the inference tree.


\section{Implementation}\label{sec:results}
This chapter describes how we implemented our verification tool for all three versions of our logic. We will also present implementations of classic algorithms and how we prove their correctness and cost using our tool.
We have implemented our verification system prototype in OCaml, and all the code and examples are in the GitHub repository \url{https://github.com/carolinafsilva/time-verification}.

\section{Tool Architecture}

Our goal is to write programs with annotation of correctness and time bounds and be able to prove these conditions. In figure~\ref{fig:architecture} we show the architecture of our tool, with each of the steps that will allow us to meet our goal.

\begin{figure}[htbp]
  \centering
  \includegraphics[width=0.8\linewidth]{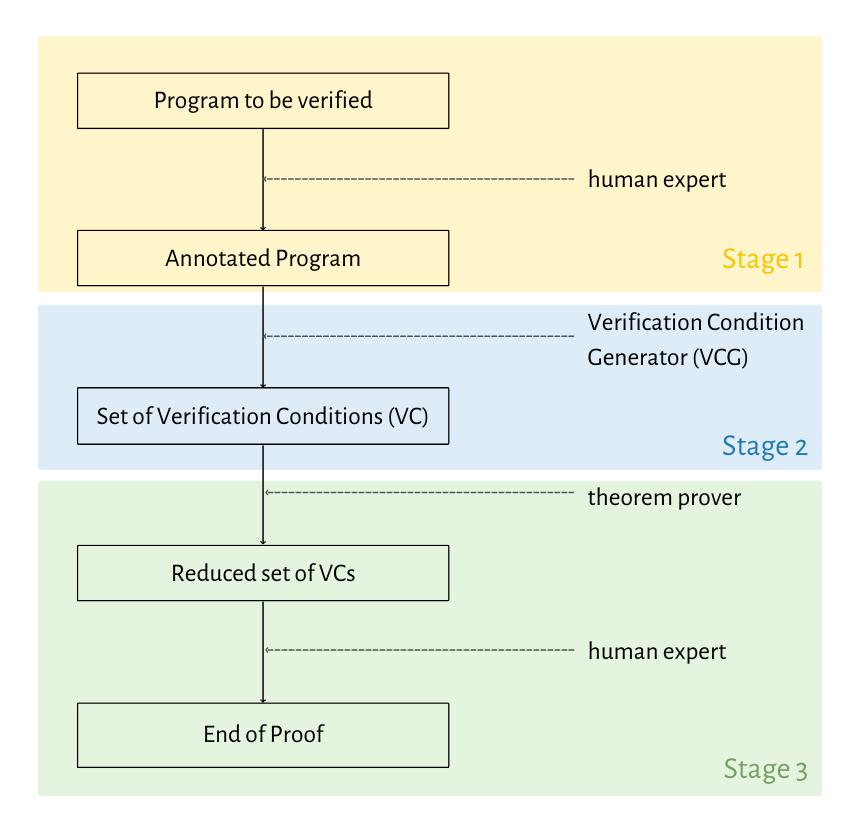}
  \caption{Architecture of our tool.}
  \label{fig:architecture}
\end{figure}

Program verification is conducted in three stages:
\begin{enumerate}
    \item Annotation of the program by the programmer, who specifies the correctness conditions that must be met, as well as the cost upper bound.
    \item Implementation of the VCG which, given an annotated program generates a set of goals that need to be proved.
    \item The proof stage: proof goals are passed to a theorem prover which attempts to prove them automatically. If it fails, some interaction with the user is needed to guide the proof.
\end{enumerate} 

Achieving the first stage involves defining the language Abstract Syntax Tree (AST), parsing the program and annotations, and implementing the operational semantics interpreter.

The second stage includes the implementation of our VCG algorithm and interaction with the oracle, here instantiated with user annotations, to provide extra information about program loops. At the end of stage two, our tool has generated a set of proof goals needed to ensure correction, termination, and resource usage of our input program. 

Finally, in stage three, we discard these proof goals by sending them to an automatic prover (e.g., Easycrypt, why3) for validation. This step might need some assistance from the user since the generated VCs might be too complex to be automatically proved. If all our VCs are validated, we know our program is correct, and we have learned some restrictions on its execution time. 

Verification is semi-automatic in the sense that in certain situations, the user has to give extra information to the program, either in the form of an oracle that defines some needed parameters or in the proof stage in situations where the proof is interactive.

\section{Implementation Details}

\subsection{Cost Model}
One essential part of our system is our cost model. We need to define a way to evaluate the cost of a program. We start by defining a map to store the cost of atomic operations. For instance, the cost of a sum ($C_{+}$) might be defined as 1, and the cost of multiplication ($C_{*}$) as 3. Besides our dictionary, we implement a semantic to define the cost of evaluating arithmetic expressions, boolean expressions, and statements. Our operational semantics uses the cost model to compute the real execution cost.

\subsection{Oracle}

Programs with loops require additional information to prove correctness, termination, and cost bounds. This information is provided through an oracle. This oracle will request user input when necessary to complete the VCG algorithm. Since information such as invariants is needed multiple times throughout the algorithm, we need to store this information to be easily accessible. To achieve this, we assign a unique identifier to each \textit{while-loop} and create an oracle hashtable to store the oracle information for each loop. The code for our oracle can be seen in listing~\ref{lis:oracle}.

\begin{lstlisting}[language=ml,label=lis:oracle,caption={Oracle Implementation.}]
let oracle_hashtbl = Hashtbl.create 43

let parse_info name parser_function =
  Printf.printf name ;
  let input = read_line () in
  Lexing.from_string input |> parser_function Lexer.token

let oracle () =
  try
    let inv = parse_info "Invariant: " Parser.annot_start in
    let f = parse_info "Variant: " Parser.aexp_start in
    let n = parse_info "Number of iterations: " Parser.aexp_start in
    let t = parse_info "Cost Function of While: " Parser.lambda_start in
    (inv, f, n, t)
  with _ -> failwith "Oracle Error\n"

let get_oracle id =
  if Hashtbl.mem oracle_hashtbl id then Hashtbl.find oracle_hashtbl id
  else
    let inv, f, n, t = oracle () in
    Hashtbl.add oracle_hashtbl id (inv, f, n, t) ;
    (inv, f, n, t)
\end{lstlisting}

\subsection{VCG}
Let us look at our VCG implementation. We implemented this algorithm as similar to the theoretical definitions as possible to guarantee soundness, as per our proofs. We presented three theoretical definitions of our logic, one classic for upper bounds, one which uses amortized analysis to further refine our upper-bound estimation, and one that proves the exact costs of a restricted version of our language.

Consider the implementation of the wp function in listing~\ref{lis:wp}. 
Note how in the while case we start by calling the $get\_oracle$ function to get our loop information.

\begin{lstlisting}[language=ml,label=lis:wp,float=htb,caption={Weakest Precondition Implementation.}]
let rec wpc s phi =
  match s with
  | Skip ->
      (phi, Var "Skip")
  | Assign (x, a) ->
      (subst phi x a, Sum (Var "Assign", time_aexp a))
  | ArrAssign (x, a1, a2) ->
      let t' = Sum (time_aexp a1, time_aexp a2) in
      (subst_arr phi x a1 a2, Sum (Var "Assign", t'))
  | Seq (s1, s2) ->
      let phi', t2 = wpc s2 phi in
      let phi, t1 = wpc s1 phi' in
      (phi, Sum (t1, t2))
  | If (b, s1, s2) ->
      let wp1, t1 = wpc s1 phi in
      let wp2, t2 = wpc s2 phi in
      let v_b = annot_of_bexp b in
      let tb = time_bexp b in
      (AAnd (AImpl (v_b, wp1), AImpl (ANeg v_b, wp2)), Sum (Sum (t1, t2), tb))
  | While (id, b, _) ->
      let inv, f, n, t = get_oracle id in
      let time =
        Sum
          ( Mul (Sum (n, Cons 1), time_bexp b)
          , Sigma ("k", 0, Sub (n, Cons 1), lambda_app t (Var "k")) )
      in
      (AAnd (inv, AGe (f, Cons 0)), time)
\end{lstlisting}

Similarly, we can see the VC function implementation in listing~\ref{lis:vc}. The \textit{while} case calls the $get\_oracle$ function again to retrieve the information about loop invariant, variant, and cost.

\begin{lstlisting}[language=ml,label=lis:vc,float=htb,caption={VC Implementation.}]
let rec vc s phi: annot list =
  match s with
  | Skip | Assign (_, _) | ArrDef (_, _) | ArrAssign (_, _, _) ->
      [] 
  | Seq (s1, s2) ->
      vc s1 (wp s2 phi) @ vc s2 phi
  | If (_, s1, s2) ->
      vc s1 phi @ vc s2 phi
  | While (id, b, s') ->
      let inv, f, n, t = get_oracle id in
      let b = annot_of_bexp b in
      let wp, t' = wpc s' (AAnd (inv, AGt (f, Var "k"))) false in
      AForall ("k", AImpl (AAnd (inv, AAnd (b, AEq (f, Var "k"))), wp))
      :: AImpl (AAnd (inv, AAnd (b, AEq (f, Var "k"))), AGe (lambda_app t (Var "k"), t'))
      :: AImpl (AAnd (inv, ANeg b), phi)
      :: AImpl (AAnd (inv, b), ALe (f, n))
      :: vc s' (AAnd (inv, ALe (f, Var "k")))
\end{lstlisting}

Finally, we show the entry point function VCG in listing~\ref{lis:vcg}. This function calls both $wpc$ and $VC$ and combines all the VCs together.

\begin{lstlisting}[language=ml,label=lis:vcg,float=htb,caption={VCG Implementation.}]
let vcg pre s t pos =
  let wp, ts = wpc s pos true in
  AImpl (pre, ALe (ts, t)) :: AImpl (pre, wp) :: vc s pos
\end{lstlisting}

The VCG implementation for amortized costs differs from the previous one, only for the \textit{while} case. The oracle will also request new information in this version, an amortized cost and a potential function instead of a function of cost.

The exact cost version of our logic requires additional changes. We start by extending our language with for-loops and implementing the required adaptations to our interpreter. A restriction will be added to \textit{if} statements to ensure equal run time for both branches. Our VCG algorithm will now be extended to deal with \textit{for}-loops and the inequality operator in the cost assertion will now be replaced with an equality operator to prove the exact-time bound.

\section{Examples}

Let us now analyze some working examples implemented in this language and the conditions generated by our VCG algorithm. Particularly we will be able to look at the implementation and results we have already analyzed in previous chapters. For simplicity, we defined the cost of all atomic operations as 1 in our cost dictionary.
In table~\ref{tab:examples-info} we show how many VCs each example generated and which of the three versions of our logic was used.

\subsection{Insertion Sort}

Our first example is of a classic sorting algorithm, insertion sort. The implementation is presented in listing~\ref{lis:insertion-sort}.

\begin{lstlisting}[language=ml,label=lis:insertion-sort,float=htb,caption={Insertion Sort Implementation with Annotations.}]
{ n > 0 }
i = 1;
while i < n do
  key = x[i];
  j = i - 1;
  while x[j] > key and j >= 0 do
    x[j + 1] = x[j];
    j = j - 1
  end;
  x[j + 1] = key;
  i = i + 1
end
{ forall k. (0<=k and k<n) => x[k] >= x[k-1] | 9*n*n + 27*n + 13 }
\end{lstlisting}

The precondition simply states that $n$ is a positive number. The postcondition says that our final array is in ascending order.
Since the implementation has two {\em while} loops, we will have two calls to the oracle. 

For the external loop, we define the maximum number of iterations as $n$. The variant is the $i$ variable since it always increases until it reaches the value of $n$. The invariant states that the array is always ordered from the first position until the $(i-1)$-th position: $\forall k. (0< k \land k < i) \rightarrow (x[k-1] \le x[k])$. 
The cost of the body of the external {\em while} is not the same for all iterations, since we have a nested while. We define this cost with the function: $t(i) = 9 \times i + 15$.

For the internal loop, it will iterate $i$ times. The variant is the increasing expression $i-j$. 
The invariant is that all elements between positions $i$ and $j$ are greater than the key and that from the first position until $i-1$ the array is sorted, excluding the element on position j: $(\forall k. (j < k \land k < i) \rightarrow x[k] > key) \rightarrow \forall k1, k2.  0 \le k1 \land k1 \le k2 \land k2 < i \land \neg (k1 = j) \land \neg (k2 = j) \rightarrow x[k1] \le x[k2]$.
Note that one can define a cost function $t(k)$ that would allow us to derive an exact cost for the {\em while} rule, however, our logic does not allow proving that this bound is tight.

Given this information, our VCG generates the conditions needed to prove the termination, correctness, and cost bound of our program. 


\subsection{Binary Search}

Our next example is another classic algorithm, Binary Search. Here we want to prove that, not only our implementation is correct and terminates, but also that the algorithm runs in logarithmic time in the size of the array. The specification can be seen in listing~\ref{lis:binary-search}

\begin{lstlisting}[language=ml,label=lis:binary-search,float=htb,caption={Binary Search Implementation with Annotations.}]
{(forall i. (0 <= i and i < n) => a[i] < a[i+1]) and (exists j. a[j] = v)}
l = 0;
u = n - 1;
while l <= u do
  m = l + ((u - l) / 2);
  if a[m] < v then
    l = m + 1
  else
    if a[m] > v then
      u = m - 1
    else
     result = m;
     l = u + 1
    end
  end
end
{0 <= result and result < n and a[result] = v | 43 * log(n) + 10}
\end{lstlisting}

The precondition states that the array $a$ is sorted, and that value $v$ is in the array. The postcondition says that $result$ is a valid position in $a$ and it corresponds to the position of $v$ in $a$, $a[result]=v$.

To prove the execution time bound, we provide to the oracle the maximum number of iterations as being $log(n)$ and a constant value as the cost of each loop body iteration.
To prove termination we must also provide $n-u+l$ as a variant. Our invariant says that the position we are looking for is between $l$ and $u$, $0 \le l \land u < n \land (\forall i. (0 \le i \land i < n \land a[i] = v) \rightarrow l \le i \land i \le u)$ as invariant.

\subsection{Binary Counter}
In the binary counter algorithm, we represent a binary number as an array of zeros and ones. We start with an array with every value at zero, and with each iteration, we increase the number by one until we reach the desired value. Our implementation can be seen in listing~\ref{lis:binary-counter}. 

\begin{lstlisting}[language=ml,label=lis:binary-counter,float=htb,caption={Binary Counter Implementation with Annotations.}]
{n >= 0 and size = log(n)}
i = 0;
while i < n do
  j = 0;
  while B[j] = 1 do
    B[j] = 0;
    j = j + 1
  end;
  B[j] = 1;
  i = i + 1
end
{n = sum(i,0,log(n) - 1, B[i]*2^i) | 20*c*n + 3*n + 30}
\end{lstlisting}

Unlike in previous examples, we applied our amortized logic to prove the bound of the binary counter algorithm.
Our precondition says that $n$ is a positive value, size is $log(n)$ and that all elements in array $B$ from 0 to $size$ start at zero. Our postcondition says that at the end of the program, array $B$ is a binary representation of decimal number $n$.
In order to prove this assertion, we must provide the oracle with the amortized cost ($2c$) and a potential function denoting the number of ones in the array at each iteration.
We must also specify the invariant $i = sum(k,0,size, B[k] * 2^k)$, the variant $i$, and the maximum number of iterations $size$, to prove correctness and termination respectively.

If we were to use a worst-case analysis on this implementation, we would get that this algorithm is $\mathcal{O}(n\ log n)$, meaning we would flip every bit ($log n$) a total of $n$ times. However, this is not the case. While the first bit ($B[0]$) does flip every iteration, the second bit($B[1]$) flips every other iteration, the third ($B[2]$) every 4th iteration, and so on. We can see a pattern where each bit $B[i]$ flips every $2^i$th iteration. This will mean that, at most, we have $2n$ bit flips, meaning our algorithm is actually $\mathcal{O}(n)$, as we successfully proved with our algorithm.
we can define a potential function as:

\subsection{Range Filter}

In our last example, we implement a simple filter where, given an array ($a$) and a range [l..u], we use an auxiliary array ($b$) to filter if the elements in $a$ are within the range, listing~\ref{lis:range-filter}.

\begin{lstlisting}[language=ml,label=lis:range-filter,float=htb,caption={Range Filter Implementation with Annotations.}]
{ 0 <= l and l < u and n >= 0}
j = 0;
for i=0 to n do
  if (l <= a[i] and a[i] <= u)
  then
    b[j] = a[i];
    j = j + 1
  else
    b[j] = b[j];
    j = j + 0
  end
end
{ forall i. (0<=i and i<n) => (l <= a[i] and a[i] <= u => exists j. B[j] = a[i] )
and (l > a[i] and a[i] > u => not (exists j. b[j] = a[i]) ) | 13*n + 10 }
\end{lstlisting}
Our pre-condition states that $l$ and $n$ are positive values, and $u$ is greater than $l$.
Our postcondition says that for every $i$ element in $a$ in the range [l..u], $i$ will also be in $B$. And for every $i$ element in $a$ not in the range [l..u], $i$ will not be in $B$. 

We provide the invariant $\forall k. (0 \le k \land k<i) \rightarrow (l \le a[k] \land a[k] \le u \rightarrow \exists k. b[j] = a[i] ) \land (l > a[k] \land a[k] > u \rightarrow \neg (\exists j. b[k] = a[i]))$ in order to prove correctness.
Using our VC generator and EasyCrypt we prove that not only is this algorithm correct, the cost we provide of $13n + 10$ is the exact cost of this program. This result allows us to conclude the time it takes to run depends only on the size of the array, and not on its values.

\begin{table}[htpb] 
    \centering
    \begin{tabular}{|c|c|c|}
    \hline
         Algorithm & Logic & Number of VCs Generated \\
         \hline
         Insertion Sort & Classic & 10 \\
         Binary Search & Classic & 6 \\
         Binary Counter & Amortized & 17 \\
         Range Filter & Exact &  5\\
         \hline
    \end{tabular}
    \caption{Logic Used and Number of VCs generated by each example.}
    \label{tab:examples-info}
\end{table}

\section{Conclusions}\label{sec:conclusion}
The topic of static cost analysis is not new. There is a lot of previous research on how to get reasonable estimations of cost or worst-case scenario costs, either by using type systems or using an axiomatic semantics. Our work continues on this but aims to produce tighter bounds than what we found so far in the literature. 

We first extended the traditional logic of worst-case cost to use amortized analysis, giving better results for programs that fit the amortized analysis scenario.

Then, we further extended our logic to a restricted version of our language, where one can prove the exact cost of execution. As far as we know, this is a novel logic. This result is rather significant if we consider the application to critical systems where the worst-case cost is not enough to guarantee all safety goals. It is also relevant if applied to cryptographic implementations, where timing leakage might be a security concern. 

\section{Future Work}\label{sec:trab}

One of the first improvements we are aiming towards is to develop a single system capturing all of the cost logics together, creating a more cohesive, powerful tool.

We primarily focused on theoretical definitions and guaranteeing a sound logic that produced reasonable bounds. Our implementation is a simple prototype that serves as proof of concept of these definitions. Therefore many improvements can be made to our tool concerning efficiency, transforming it from a conceptual tool to a practical one.

We also defined our logic using a simple language, which allowed us to focus on the cost estimation aspect without having to worry so much about language details. In the future, we want to extend our language with more features, such as functions, to improve the expressiveness of programs.
We would also like to expand the application of our logic to more extensive and complex case studies, namely cryptographic implementations. 

Our work started as an adaptation of the EasyCrypt cost logic developed in~\cite{barbosa21}. In the future, we would like to propose an extension to the EasyCrypt tool with our logic.

%
%
%
\bibliographystyle{splncs04}
\bibliography{refs}
\end{document}